\documentclass[%
 preprint,
superscriptaddress,
 amsmath,amssymb,
 aps,
 prl,
]{revtex4-1}

\usepackage[english]{babel}
\usepackage{amsmath}
\usepackage{graphicx}
\usepackage{amsthm} %
\usepackage{amsthm} %
\newtheorem{corollary}{Corollary}

\usepackage{color}
\usepackage{braket}
\usepackage{amssymb}
\usepackage{amscd}
\usepackage{wasysym}
\usepackage[utf8]{inputenc}
\usepackage[T1]{fontenc}
\usepackage{ae,aecompl}
\usepackage{hyperref}
\usepackage{bm}

\newtheorem{theorem}{Theorem}

\begin{document}

\title{Inhibiting failure spreading in complex networks}

\author{Franz Kaiser}
  \affiliation{Forschungszentrum J\"ulich, Institute for Energy and Climate Research (IEK-STE), 52428 J\"ulich, Germany}
    \affiliation{Institute for Theoretical Physics, University of Cologne, K\"oln, 50937, Germany}

\author{Vito Latora}
  \affiliation{School of Mathematical Sciences, Queen Mary University of London, London E1 4NS, UK}
    \affiliation{Dipartimento di Fisica ed Astronomia, Universit{\`a} di Catania and INFN, 95123 Catania, Italy}
    \affiliation{The Alan Turing Institute, The British Library, London NW1 2DB, UK}
\author{Dirk Witthaut}%
  \affiliation{Forschungszentrum J\"ulich, Institute for Energy and Climate Research (IEK-STE), 52428 J\"ulich, Germany}
    \affiliation{Institute for Theoretical Physics, University of Cologne, K\"oln, 50937, Germany}

\date{\today}%

\begin{abstract}
In our daily lives, we rely on the proper functioning
of supply networks, from power grids to water transmission systems. A single failure in these critical infrastructures can lead to a complete collapse through a cascading failure mechanism. Counteracting strategies are thus heavily sought after. In this article, we introduce a general framework to analyse the spreading of failures in complex networks and demonstrate that both weak and strong connections can be used to contain damages. We rigorously prove the existence of certain subgraphs, called network isolators, that can completely inhibit any failure spreading, and we show how to create such isolators in synthetic and real-world networks. The addition of selected links  can thus prevent large scale outages as demonstrated for power transmission grids.
\end{abstract}

\maketitle

Complex networked systems are subject to external perturbations,
damages or attacks with potentially catastrophic
consequences  \cite{albert2000error,Yang2017}. The
loss of even a single edge can cause a blackout in a power
grid \cite{Pour06,16redundancy}, the dieback of a biological
network \cite{Sack2008}, or the collapse of an entire ecological
network \cite{Dunne2009}. It is thus essential to understand how the
structure of a network determines its response to perturbations
and its global resilience \cite{Strogatz2001,sahasrabudhe_rescuing_2011,motter_cascade_2004,dsouza_curtailing_2017,lin_self-organization_2018}. Here, we propose a general framework to model the redistribution of flows in a complex network that follows a small and local failure, and we suggest novel and more efficient strategies to improve network resilience. Our findings reveal that propagation of damages can counterintuitively be better limited by adding selected links instead of removing links and can turn very useful to construct more robust networks or to improve existing ones.

The division of a network into weakly coupled parts provides the most intuitive method to inhibit the spreading of failures, thus improving system resilience \cite{Stouffer2011,gilarranz_effects_2017,Hens2019,brummitt_suppressing_2012}. An example is shown in Fig.~\ref{fig:motifs_preventing}(a) for an elementary supply network with two weakly connected modules. The response to an edge failure is strong locally, but it is reduced in the other module of the network which has only few links to the part where the failure happened. A similar effect is observed in a real-world case, the Scandinavian power grid in Fig.~\ref{fig:motifs_preventing}(d). The study of community structures in both natural and man-made systems is an integral part of network science: a variety of methods has been developed to define and identify the weakly connected modules of a network \cite{radicchi_defining_2004,girvan_community_2002,newman_communities_2012}, and the important role of community structures in network dynamics is today well recognised.

\begin{figure*}[tb]
    \begin{center}
    \includegraphics[width=1.0\textwidth]{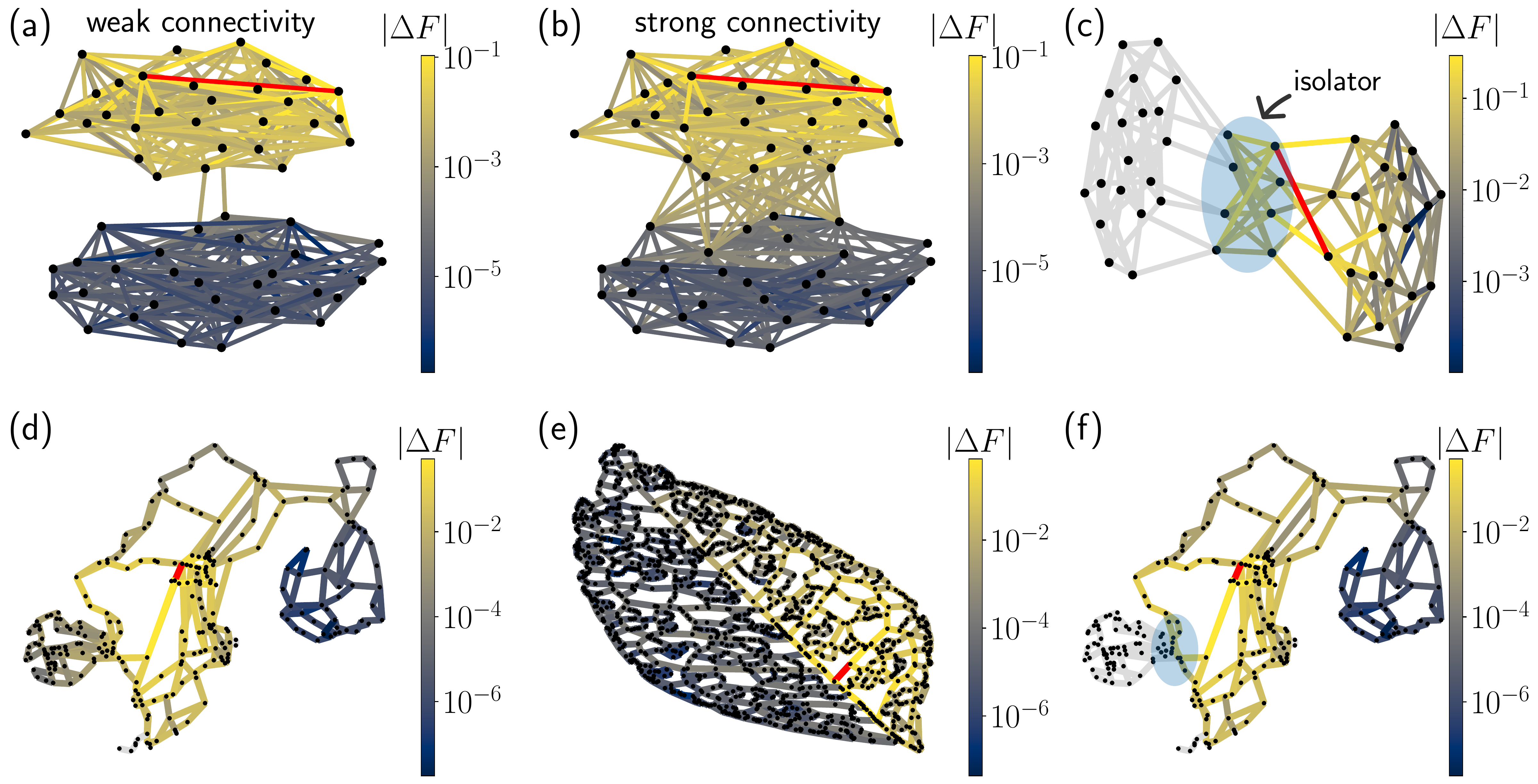}
    \end{center}
    \caption{ \textbf{Different network structures inhibit the spreading of failures in complex networks.} 
    We simulate the impact of a single failing link (red) for different network structures; resulting flow changes are colour coded. (a,b) Both a weak and a strong inter-connectivity can suppress the spreading of failures between two modules of a complex network. 
    (c) Failure spreading is prevented completely by {a network isolator} (blue shading); flow changes on the grey links are exactly zero. 
    (d) The Scandinavian powergrid consists of three weakly connected modules which suppresses failure spreading between the modules~\cite{horsch_2018}.  
    (e) The vascular network of a \textit{Bursera hollickii} leaf contains a strong central vein \cite{ronellenfitsch_topological_2015}, which suppresses failure spreading between the two sides of the leaf. 
    (f) Same as in (d) but with the addition of two links (blue shading) to create a network isolator.
    \label{fig:motifs_preventing} 
    }
\end{figure*}

Limiting connectivity for the sake of additional security is, however, not always desirable. For instance, microgrid concepts and intentional islanding are heavily discussed in energy systems research \cite{Lasseter2002, mureddu2016islanding}, but the overall demand for electric power transmission actually increases \cite{Schlachtberger2017,TRONDLE2020}. Other methods to contain perturbations or damages in complex networks are thus needed. Indeed, an exceptionally strong inter-connectivity between two modules can also suppress failure spreading as shown in Fig.~\ref{fig:motifs_preventing}(b,e). Notably, a strong inter-connectivity can be realised in different ways. In the random network example in Fig.~\ref{fig:motifs_preventing}(b), a high number of links connects a subset of nodes of the two modules. In real vascular networks of leafs, the suppression of failure spreading occurs naturally because the central vein between the left and right parts has an exceptionally large capacity (Fig.~\ref{fig:motifs_preventing}(e), cf. also~ \cite{gavrilchenko_resilience_2019}). 

Remarkably, failure spreading can be completely stopped by certain subgraphs which we refer to as {\em network isolators} in the following, an example being shown in Fig.~\ref{fig:motifs_preventing}(c). The failure of an edge in the right part of the network does not affect the flows in the left part at all. Real world networks can be made perfectly resistant to failure propagation by the ad-hoc addition of few links to create network isolators, as demonstrated for the Scandinavian power grid in Fig.~\ref{fig:motifs_preventing}(f) consisting of three weakly coupled modules. The suppression of failure spreading is readily generalised to networks with more than two modules~\footnote{See Supplemental Material for a proof of the main theorem on network isolators, extended Figures on their applicability, a discussion of their relationship to network controllability and detailed methods, which include Refs.~\cite{newman2010,bollobas1998,Dorfler2018,Wood14,motter_ml_2002,crucitti_model_2004,yuan_exact_2013,van_mieghem_2017,gao_target_2014,liu_controllability_2011,Norman97,Rohden2012,nishikawa2015,manik_supply_2014,rodrigues_kuramoto_2016,wurbs2001,reichold_vascular_2009}
}. 

Our results are based on a general framework that allows a theoretical analysis of the interplay of network connectivity and robustness in the context of supply or transportation networks. Many such systems can in fact be modelled by linear flow networks where the flow over an edge $(i,j)$ depends linearly on the gradient of a potential function across the edge, $F_{i \rightarrow j} = K_{ij} \cdot(\vartheta_i - \vartheta_j)$. In particular, this description applies to power transmission grids \cite{Purc05}, where $F$ is the real power flow, $\vartheta_i$ denotes the nodal voltage phase angle and $K_{ij}$ is given by the line susceptance, and to vascular networks \cite{Kati10,coomes_scaling_2008}, where $F$ is the flow of water or nutrients, $\vartheta_i$ is the local pressure and $K_{ij}$ the edge's capacity (see Supplemental Material \cite{Note1}). The generalisation to nonlinear systems will be discussed below. Furthermore, equivalent problems arise in the linearisation of general diffusively coupled networks of dynamical systems around an equilibrium or limit cycle~ \cite{Manik2017}.

The impact of a damage in linear flow networks can be calculated analytically. Assume that an edge $\ell=(r,s)$ fails, and summarise the response at all nodes $i=1,\ldots,N$ by the vector of changes $\Delta \vec \vartheta = (\Delta \vartheta_1, \ldots,\Delta \vartheta_N)^\top$. One then finds \cite{Note1}
\begin{equation}
    \bm{L} \, \Delta \vec \vartheta = q_{\ell} \vec \nu_{\ell},
    \label{eq:dipole1}
\end{equation}
where $\bm{L}$ is the Laplacian matrix of the weighted network, $\vec \nu_{\ell}$ is a vector with  $+1$ at position $r$ and $-1$ at position $s$, and $q_{\ell} = 1 - K_{rs} \vec \nu_{\ell}^\top \bm{L}^{-1} \vec \nu_{\ell}$ is a source strength~\mbox{ \cite{strake2018}}. We thus find that the response of a network to failures is essentially determined by the Laplacian $\bm{L}$. This matrix incorporates the network structure and is defined as follows;
\begin{equation}
    L_{ij}   =\left\{\begin{array}{l l }
      -K_{ij} & \; \mbox{if $i$ is connected to $j$},  \\
      \sum_{(i,k)\in E(G)}K_{ik} & \; \mbox{if $i=j$},  \\
      0     & \; \mbox{otherwise}.
  \end{array} \right. \label{eq:Laplacian}
\end{equation}
Notably, the Laplacian matrix is also useful to infer the large scale connectivity and the community structure of a given network~ \cite{fortunato_community_2010}. 

To quantify the effect of connectivity on failure spreading, we have studied the impact of different failures in a variety of synthetic networks as well as in several real-world networks. For a given initial failure of an edge $\ell$, we compute the flow changes $\Delta F_{i\rightarrow j}=K_{ij}\cdot(\Delta \vartheta_i-\Delta\vartheta_j)$ for all edges $(i,j)$ in a given subgraph $G'$ of the network. Furthermore, we must take into account that the impact of a failure generally decreases with distance~ \cite{strake2018,kaiser_collective_2020,Kett15}.  As an overall measure of the impact of a failure we thus consider the expression $\langle|\Delta F_{i\rightarrow j} | \rangle_d^{(i,j)\in G'}$, which gives the magnitude of flow changes averaged over all edges $(i,j) \in G'$ at a given distance $d$ to the edge $\ell$. The prime question is now whether the impact differs substantially between the communities or moduli of a network. We thus plot the ratio
\begin{equation}
R(\ell,d)=\frac{\langle|\Delta F_{i\rightarrow j} | \rangle_d^{(i,j)\in \text{O}}}{\langle
  | \Delta F_{i\rightarrow j} |
  \rangle_d^{(i,j)\in \text{S}}} \, .
\label{eq:flowratio}
\end{equation}
between the remote part of the network $G'=\text{O}$ and the part $G'=\text{S}$ containing the failing edge $\ell$. If this ratio approaches or reaches zero, this is indicative of a very strong suppression of failure spreading into the other part of the network.

\begin{figure}[tb]
    \begin{center}
    \includegraphics[width=0.7\columnwidth]{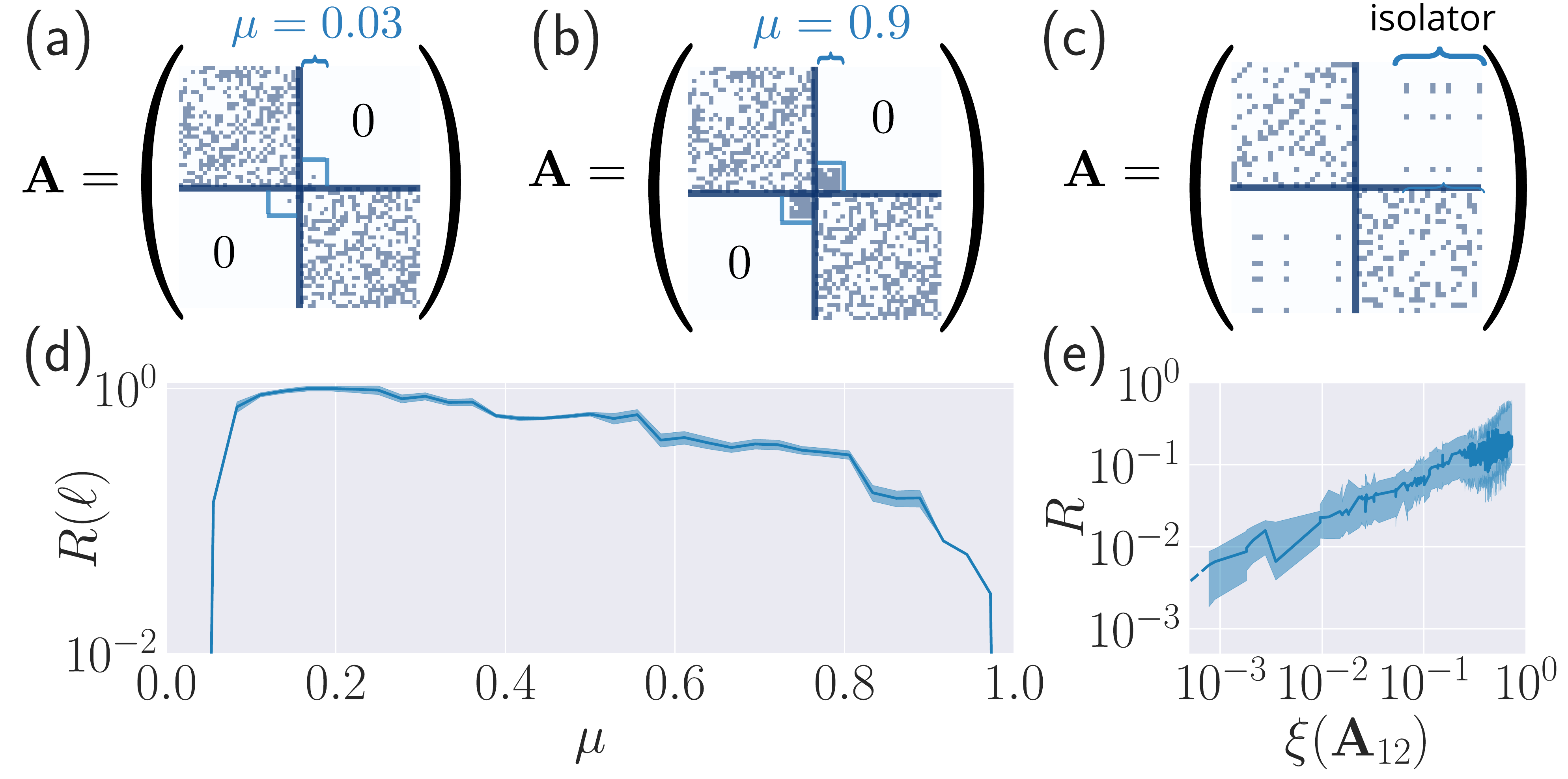}
    \end{center}
    \caption{
    \textbf{Effectiveness and robustness of shielding network structures.}
    (a,b) Adjacency matrices for the graphs shown in Fig.~\ref{fig:motifs_preventing}(a,b). Two random graphs $G(30,0.4)$ are inter-connected via a fraction $c$ of their nodes chosen at random, and links are added with probability $\mu$, interpolating between weak (a) or strong (b) inter-connectivity. 
    (c) Adjacency matrix for the 6-regular graph shown in Fig.~\ref{fig:motifs_preventing}(c) containing a network isolator.
    (d) The average ratio of flow changes $R(\ell)$ in the two components (Eq.~\ref{eq:flowratio}) is strongly suppressed for both high and low inter-connectivity $\mu$. The blue line represents the median value over all distances and the shaded region indicates the $0.25$- and $0.75$-quantiles.
    (e) The ratio of flow changes $R$ vanishes for a perfect network isolator described by the condition $\xi(\mathbf{A}_{12})=0$ and increases algebraically with coherence parameter $\xi$ when perturbed.
    \label{fig:motifs_robustness}
    }
\end{figure}

\begin{figure}[tb]
    \begin{center}
        \includegraphics[width=0.7\columnwidth]{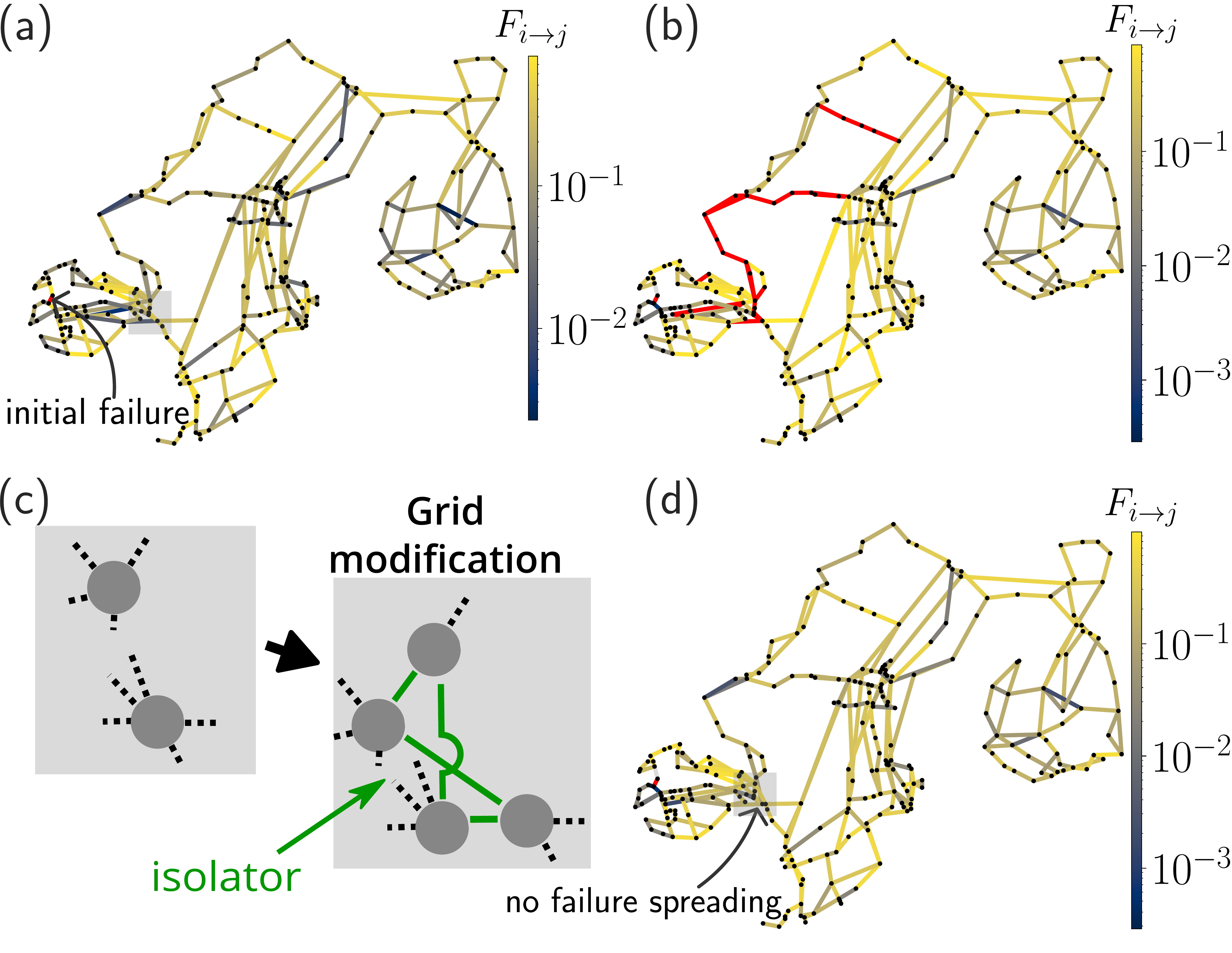}
    \end{center}
    \caption{
    \textbf{Network isolators can contain cascading failures in power grids.}
    (a) Line loading (colour code) on the Scandinavian grid in units relative to maximal loading before the initial failure of a single line (coloured red).
    (b) The initial failure results in a cascade of overloads (red coloured lines) until the grid disconnects into several parts. 
   (c) Magnification of the grid structure in Eastern Norway. A small modification of the grid enables a network isolator.
   (d) Introducing the network isolator completely suppresses the spreading of failures from Western Norway to the rest of the grid thus inhibiting the cascade observed in (b). 
    \label{fig:gridisolator}
    }
\end{figure}

To study how the impact of failure spreading depends on the network structure, we considered synthetic graphs obtained by connecting two Erd\H{o}s-R\'{e}nyi (ER) random graphs to each other at preselected, randomly chosen vertices with a tunable probability $\mu\in [0,1]$~\cite{Erdos1960,Note1}. The resulting graphs have a connectivity structure ranging from two weakly connected communities for low values of $\mu$ shown in Figure~\ref{fig:motifs_preventing}(a) to strongly connected parts shown in panel (b). In the limit $\mu=1$, the two modules are connected via a complete bipartite graph as shown in Fig.~\ref{fig:motifs_preventing}(c). This is a possible realization of a \textit{network isolator}, since it completely suppresses flow changes as we will explain in the following. The corresponding adjacency matrices clearly indicate the connectivity structure, revealing the strong or weak coupling between the two modules of the networks (Fig.~\ref{fig:motifs_robustness}(a,b,c)). Remarkably, evaluating the quantity $R(\ell)$, obtained by averaging the ratio over flow changes $R(\ell,d)$ over all distances $d$ for a specific trigger link $\ell$, for a varying connectivity structure tuned by $\mu$, we find that the spreading of failures is largely suppressed for both weak and strong connectivity between the two modules as shown in Fig.~\ref{fig:motifs_robustness}(d). Distance plays a minor role for the ratio of flow changes $R(d)$ as illustrated in a separate Figure in the Supplemental Material~\cite{Note1}.

Network symmetries are known to play an important role for the dynamics and synchronisability of a network~ \cite{pecora_cluster_2014,sorrentino_complete_2016,nicosia_remote_2013}. Network isolators as a specific connectivity structure completely inhibit the spreading of failures from one network module to another. They manifest also as particular, symmetric patterns in the region of the adjacency matrix describing the connectivity between the two parts of the network as we have seen in Fig.~\ref{fig:motifs_robustness}(c). They are characterised by the following theorem, which we prove in the Supplemental Material~\cite{Note1}
\begin{theorem}
Consider a linear flow network composed of two modules 1,2 and let $\mathbf{A}_{12}$ denote the weighted adjacency matrix of the mutual connections. An edge failure in one module does not affect the flows in the other module if $\operatorname{rank}(\mathbf{A}_{12}) = 1$. For unweighted networks this criterion is fulfilled if $\mathbf{A}_{12}$ describes a complete bipartite graph.
\end{theorem}
Note that, while network isolators prevent failure spreading, we found that they do not influence network controllability~\cite{Note1}.

\begin{figure*}[h!]
    \begin{center}
    \includegraphics[width=1.0\textwidth]{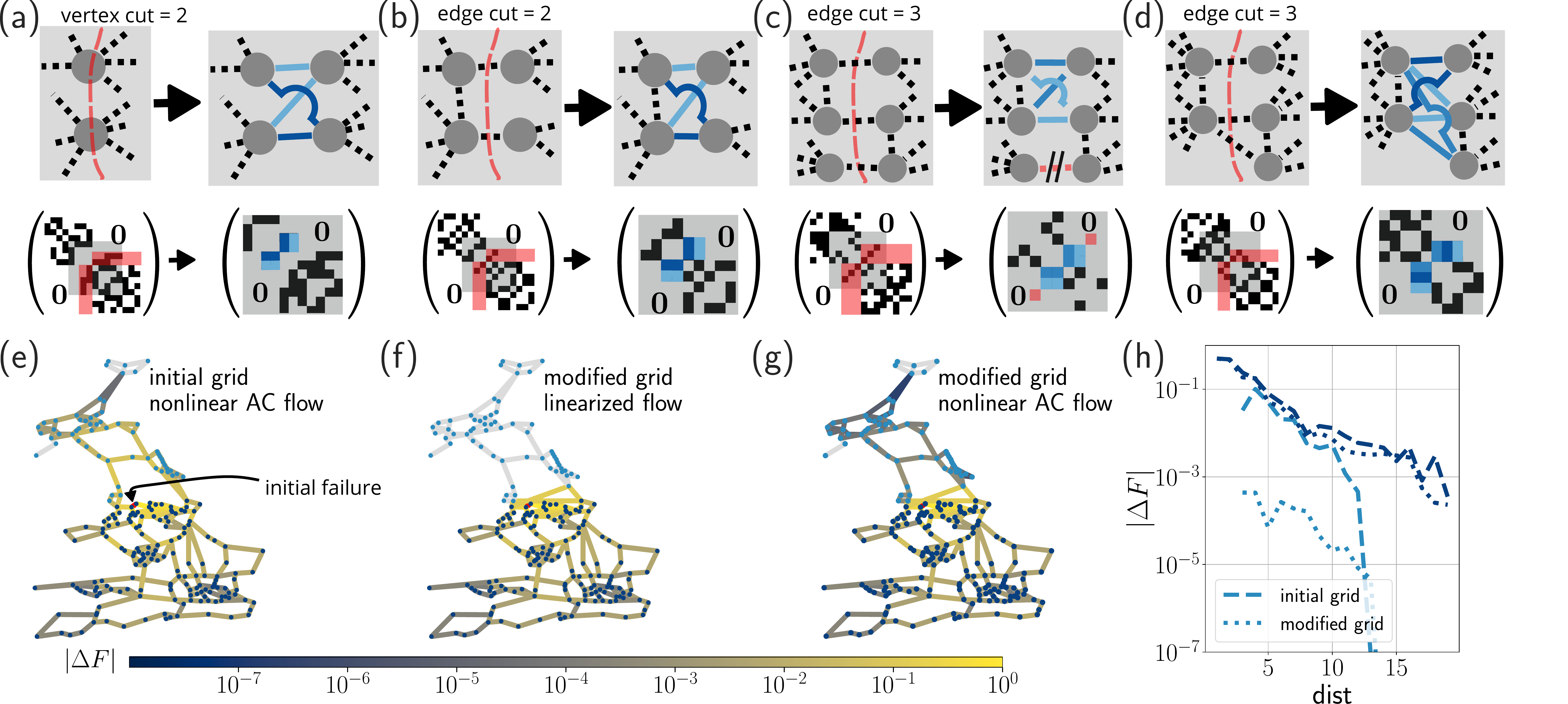}
    \end{center}
    \caption{ 
    {\textbf{Different ways of constructing isolators and their  effects in full non-linear AC load flow.}
    (a) to (d) Alternative methods of creating an isolator in a given network. We show the network structure before (top left) and after (top right) the addition of a network isolator, as well as corresponding adjacency matrices (bottom) with the different shades of blue representing the capacity $K_{ij}$ of the respective edge.
    A lower prior connectivity simplifies the creation of isolators as measured by the vertex cut (a) or edge cut (b,c,d)
    which is visible in the adjacency matrix (entries colored red). The creation of network isolators results in characteristic patterns in the adjacency matrix in terms of the capacities of the isolator edges (shades of blue).
    (e) An initially failing link with unit flow (red) in the British grid results in changes of real power flow (colour code) throughout the whole network, as obtained by computing a non-linear full AC power flow~\cite{Note1,horsch_2018}. (f,g) After introducing a network isolator based on the strategy presented in panel (a), failure spreading is perfectly inhibited in the linear power flow approximation, and still significantly reduced in the non-linear full AC load flow. (h) Median absolute flow changes are significantly lower in the upper module of the grid shielded by the isolator (light blue, straight line) compared to the situation before introducing the isolator (light blue, dotted line), whereas flow changes in the lower module are almost the same (dark blue lines).
    }
    \label{fig:recipes_and_nonlinear_isolator}
    }
\end{figure*}

Since most real world examples of networks do not contain perfect network isolators, we have studied the robustness of a network isolator against modifications of the topology.  Starting from a unit rank matrix, we perturb the adjacency matrix $\mathbf{A}_{12}$ iteratively. In each step we choose one of the matrix columns $\vec a_i, i=1,\ldots,m$ at random and perturb it according to $\vec{a}_i^\prime = \vec{a}_i+ \vec{e} \lVert\vec{a}_i\rVert$. The elements of the perturbation vector $\vec{e}$ are chosen uniformly at random from the interval $[-\beta,\beta]$, where $\beta$ is a small parameter, here $\beta=0.05$. The deviation of the perturbed matrix $\mathbf{A}_{12}$ from a unit rank matrix is quantified using its coherence statistics~\cite{tropp_computational_2010}, 
\begin{equation}
    \xi(\mathbf{A}_{12}) = 1-\min_{i,j}\frac{\langle \vec{a}_i,\vec{a}_j\rangle}{\lVert\vec{a}_i\rVert\lVert\vec{a}_j\rVert} \, .
\end{equation}
The performance of the isolator is then measured by calculating the ratio of flow changes $R$, which is obtained from $R(\ell,d)$ by averaging over all possible trigger links and distances. A perfect isolator is characterised by  $\xi(\mathbf{A}_{12})=0$ and enables a complete containment of failure spreading such that $R=0$. For perturbed isolators, we find that $R$ increases approximately algebraically with $\xi(\mathbf{A}_{12})$, see Fig.~\ref{fig:motifs_robustness}(e). Hence, the isolation effect persists for small perturbations, albeit with reduced efficiency. 

Perfect network isolators can be easily constructed to improve the resilience of complex networked systems. As a practical example we show an application to electric power grids, where large scale blackouts are typically triggered by the outage of a single transmission element which leads to a cascade of failures \cite{Pour06,Schafer2018}. We demonstrate the impact of network isolators against cascading failures in the case of the Scandinavian grid. In the original grid layout, failure spreading between different areas of the grid is reduced due to the presence of only few connections between different areas of the network -- but it is possible. A failure in one area can thus spread to other areas and cause a global cascade, as demonstrated in Figure \ref{fig:gridisolator}(a,b) for a cascade emerging in Western Norway. This spreading may in principle be prevented by decoupling different areas of the grid, but this is highly undesired. In fact, future energy systems will require more connectivity, not less, to transmit renewable electric power \cite{Schlachtberger2017,TRONDLE2020}. In contrast, building a network isolator can completely inhibit failure spreading at increased connectivity. A perfect isolator can be realised with moderate effort by reconstructing two substations in Norway, such that they effectively form two nodes each. The new nodes must be linked by internal connections and one additional two-circuit overhead line, whose parameters are optimised such that the condition $\operatorname{rank}(\mathbf{A}_{12}) = 1$ is satisfied (Fig.~\ref{fig:gridisolator}(c)). A simulation for such an optimised grid layout shows that the spreading of the cascade is completely suppressed (Fig.~\ref{fig:gridisolator}(d)). The grid remains connected and load shedding is no longer necessary as a containment strategy \cite{Pour06,Yang2017}.

Network isolators are not limited to the particular situation shown in Figure~\ref{fig:motifs_preventing}, or to linear flow models only. In Figure~\ref{fig:recipes_and_nonlinear_isolator}(a,d) we identify several subgraphs that allow to easily introduce network isolators into existing topologies. For subgraphs with a prior low connectivity, as measured by a small vertex cut (Fig.~\ref{fig:recipes_and_nonlinear_isolator}(a)) or a small edge cut (Fig.~\ref{fig:recipes_and_nonlinear_isolator}(b,c,d)), network isolators may be introduced with small grid modifications - by adding (a,b,d) or removing and adding (c) selected links. The concept of network isolators has been established for linear flow networks, but can be generalised in two ways. (1) We can rigorously proof the existence of isolators for an important class of nonlinear network dynamical systems \cite{Note1}. (2) For many nonlinear systems of practical importance, the impact of failures or perturbations is well described by a linearisation around an equilibrium or limit cycle \cite{Manik2017}. We demonstrate that the strong isolation of line outages in power grids persists beyond the linearised flow equations via direct numerical simulations in Fig.~\ref{fig:recipes_and_nonlinear_isolator}(e to h).

In conclusion, connectivity determines the resilience of complex networks in manifold ways. As expected, a division of a network into weakly coupled modules suppresses the spreading of failures from one module to the others. Remarkably, we have found that a strong inter-connection can equally well suppresses the spreading in both flow networks and in networks of nonlinear dynamical systems. We have demonstrated that an even stronger effect can be created by certain subgraphs called isolators, which inhibit the spreading completely. Isolators can be applied to mitigate cascading failures, for instance in electric power grids, while enabling an arbitrary degree of connectivity between the grid parts. These results widen our perspective on the large scale organisation of complex networks in general, showing that very diverse structural patterns can exist that discriminate functional units and improve network resilience. 

We thank Tom Brown and Jonas H\"orsch for help with the processing of power grid data, Eleni Katifori and Henrik Ronellenfitsch for providing the leaf data and Raissa D'Souza and Jürgen Kurths for valuable discussions. We gratefully acknowledge support from the Federal Ministry of Education and Research (BMBF grant no. 03EK3055B D.W.), the Helmholtz Association (via the joint initiative "Energy System 2050 -- a Contribution of the Research Field Energy" and grant no. VH-NG-1025 to D.W.). V.L.'s work was funded by the Leverhulme Trust Research Fellowship CREATE.

\bibliography{bibliography}

\begin{thebibliography}{57}%
\makeatletter
\providecommand \@ifxundefined [1]{%
 \@ifx{#1\undefined}
}%
\providecommand \@ifnum [1]{%
 \ifnum #1\expandafter \@firstoftwo
 \else \expandafter \@secondoftwo
 \fi
}%
\providecommand \@ifx [1]{%
 \ifx #1\expandafter \@firstoftwo
 \else \expandafter \@secondoftwo
 \fi
}%
\providecommand \natexlab [1]{#1}%
\providecommand \enquote  [1]{``#1''}%
\providecommand \bibnamefont  [1]{#1}%
\providecommand \bibfnamefont [1]{#1}%
\providecommand \citenamefont [1]{#1}%
\providecommand \href@noop [0]{\@secondoftwo}%
\providecommand \href [0]{\begingroup \@sanitize@url \@href}%
\providecommand \@href[1]{\@@startlink{#1}\@@href}%
\providecommand \@@href[1]{\endgroup#1\@@endlink}%
\providecommand \@sanitize@url [0]{\catcode `\\12\catcode `\$12\catcode
  `\&12\catcode `\#12\catcode `\^12\catcode `\_12\catcode `\%12\relax}%
\providecommand \@@startlink[1]{}%
\providecommand \@@endlink[0]{}%
\providecommand \url  [0]{\begingroup\@sanitize@url \@url }%
\providecommand \@url [1]{\endgroup\@href {#1}{\urlprefix }}%
\providecommand \urlprefix  [0]{URL }%
\providecommand \Eprint [0]{\href }%
\providecommand \doibase [0]{https://doi.org/}%
\providecommand \selectlanguage [0]{\@gobble}%
\providecommand \bibinfo  [0]{\@secondoftwo}%
\providecommand \bibfield  [0]{\@secondoftwo}%
\providecommand \translation [1]{[#1]}%
\providecommand \BibitemOpen [0]{}%
\providecommand \bibitemStop [0]{}%
\providecommand \bibitemNoStop [0]{.\EOS\space}%
\providecommand \EOS [0]{\spacefactor3000\relax}%
\providecommand \BibitemShut  [1]{\csname bibitem#1\endcsname}%
\let\auto@bib@innerbib\@empty
%</preamble>
\bibitem [{\citenamefont {Albert}\ \emph {et~al.}(2000)\citenamefont {Albert},
  \citenamefont {Jeong},\ and\ \citenamefont {Barab{\'a}si}}]{albert2000error}%
  \BibitemOpen
  \bibfield  {author} {\bibinfo {author} {\bibfnamefont {R.}~\bibnamefont
  {Albert}}, \bibinfo {author} {\bibfnamefont {H.}~\bibnamefont {Jeong}},\ and\
  \bibinfo {author} {\bibfnamefont {A.-L.}\ \bibnamefont {Barab{\'a}si}},\
  }\bibfield  {title} {\bibinfo {title} {Error and attack tolerance of complex
  networks},\ }\href@noop {} {\bibfield  {journal} {\bibinfo  {journal}
  {Nature}\ }\textbf {\bibinfo {volume} {406}},\ \bibinfo {pages} {378}
  (\bibinfo {year} {2000})}\BibitemShut {NoStop}%
\bibitem [{\citenamefont {Yang}\ \emph {et~al.}(2017)\citenamefont {Yang},
  \citenamefont {Nishikawa},\ and\ \citenamefont {Motter}}]{Yang2017}%
  \BibitemOpen
  \bibfield  {author} {\bibinfo {author} {\bibfnamefont {Y.}~\bibnamefont
  {Yang}}, \bibinfo {author} {\bibfnamefont {T.}~\bibnamefont {Nishikawa}},\
  and\ \bibinfo {author} {\bibfnamefont {A.~E.}\ \bibnamefont {Motter}},\
  }\bibfield  {title} {\bibinfo {title} {Small vulnerable sets determine large
  network cascades in power grids},\ }\href@noop {} {\bibfield  {journal}
  {\bibinfo  {journal} {Science}\ }\textbf {\bibinfo {volume} {358}},\ \bibinfo
  {pages} {eaan3184} (\bibinfo {year} {2017})}\BibitemShut {NoStop}%
\bibitem [{\citenamefont {Pourbeik}\ \emph {et~al.}(2006)\citenamefont
  {Pourbeik}, \citenamefont {Kundur},\ and\ \citenamefont {Taylor}}]{Pour06}%
  \BibitemOpen
  \bibfield  {author} {\bibinfo {author} {\bibfnamefont {P.}~\bibnamefont
  {Pourbeik}}, \bibinfo {author} {\bibfnamefont {P.}~\bibnamefont {Kundur}},\
  and\ \bibinfo {author} {\bibfnamefont {C.}~\bibnamefont {Taylor}},\
  }\bibfield  {title} {\bibinfo {title} {The anatomy of a power grid blackout -
  root causes and dynamics of recent major blackouts},\ }\href
  {https://doi.org/10.1109/MPAE.2006.1687814} {\bibfield  {journal} {\bibinfo
  {journal} {IEEE Power and Energy Magazine}\ }\textbf {\bibinfo {volume}
  {4}},\ \bibinfo {pages} {22} (\bibinfo {year} {2006})}\BibitemShut {NoStop}%
\bibitem [{\citenamefont {Witthaut}\ \emph {et~al.}(2016)\citenamefont
  {Witthaut}, \citenamefont {Rohden}, \citenamefont {Zhang}, \citenamefont
  {Hallerberg},\ and\ \citenamefont {Timme}}]{16redundancy}%
  \BibitemOpen
  \bibfield  {author} {\bibinfo {author} {\bibfnamefont {D.}~\bibnamefont
  {Witthaut}}, \bibinfo {author} {\bibfnamefont {M.}~\bibnamefont {Rohden}},
  \bibinfo {author} {\bibfnamefont {X.}~\bibnamefont {Zhang}}, \bibinfo
  {author} {\bibfnamefont {S.}~\bibnamefont {Hallerberg}},\ and\ \bibinfo
  {author} {\bibfnamefont {M.}~\bibnamefont {Timme}},\ }\bibfield  {title}
  {\bibinfo {title} {Critical links and nonlocal rerouting in complex supply
  networks},\ }\href {https://doi.org/10.1103/PhysRevLett.116.138701}
  {\bibfield  {journal} {\bibinfo  {journal} {Physical Review Letters}\
  }\textbf {\bibinfo {volume} {116}},\ \bibinfo {pages} {138701} (\bibinfo
  {year} {2016})}\BibitemShut {NoStop}%
\bibitem [{\citenamefont {Sack}\ \emph {et~al.}(2008)\citenamefont {Sack},
  \citenamefont {Dietrich}, \citenamefont {Streeter}, \citenamefont
  {Sanchez-Gomez},\ and\ \citenamefont {Holbrook}}]{Sack2008}%
  \BibitemOpen
  \bibfield  {author} {\bibinfo {author} {\bibfnamefont {L.}~\bibnamefont
  {Sack}}, \bibinfo {author} {\bibfnamefont {E.~M.}\ \bibnamefont {Dietrich}},
  \bibinfo {author} {\bibfnamefont {C.~M.}\ \bibnamefont {Streeter}}, \bibinfo
  {author} {\bibfnamefont {D.}~\bibnamefont {Sanchez-Gomez}},\ and\ \bibinfo
  {author} {\bibfnamefont {N.~M.}\ \bibnamefont {Holbrook}},\ }\bibfield
  {title} {\bibinfo {title} {{Leaf palmate venation and vascular redundancy
  confer tolerance of hydraulic disruption}},\ }\href
  {https://doi.org/10.1073/pnas.0709333105} {\bibfield  {journal} {\bibinfo
  {journal} {Proceedings of the National Academy of Sciences}\ }\textbf
  {\bibinfo {volume} {105}},\ \bibinfo {pages} {1567} (\bibinfo {year}
  {2008})}\BibitemShut {NoStop}%
\bibitem [{\citenamefont {Dunne}\ and\ \citenamefont
  {Williams}(2009)}]{Dunne2009}%
  \BibitemOpen
  \bibfield  {author} {\bibinfo {author} {\bibfnamefont {J.~A.}\ \bibnamefont
  {Dunne}}\ and\ \bibinfo {author} {\bibfnamefont {R.~J.}\ \bibnamefont
  {Williams}},\ }\bibfield  {title} {\bibinfo {title} {Cascading extinctions
  and community collapse in model food webs},\ }\href@noop {} {\bibfield
  {journal} {\bibinfo  {journal} {Philosophical Transactions of the Royal
  Society B: Biological Sciences}\ }\textbf {\bibinfo {volume} {364}},\
  \bibinfo {pages} {1711} (\bibinfo {year} {2009})}\BibitemShut {NoStop}%
\bibitem [{\citenamefont {Strogatz}(2001)}]{Strogatz2001}%
  \BibitemOpen
  \bibfield  {author} {\bibinfo {author} {\bibfnamefont {S.~H.}\ \bibnamefont
  {Strogatz}},\ }\bibfield  {title} {\bibinfo {title} {Exploring complex
  networks},\ }\href@noop {} {\bibfield  {journal} {\bibinfo  {journal}
  {Nature}\ }\textbf {\bibinfo {volume} {410}},\ \bibinfo {pages} {268}
  (\bibinfo {year} {2001})}\BibitemShut {NoStop}%
\bibitem [{\citenamefont {Sahasrabudhe}\ and\ \citenamefont
  {Motter}(2011)}]{sahasrabudhe_rescuing_2011}%
  \BibitemOpen
  \bibfield  {author} {\bibinfo {author} {\bibfnamefont {S.}~\bibnamefont
  {Sahasrabudhe}}\ and\ \bibinfo {author} {\bibfnamefont {A.~E.}\ \bibnamefont
  {Motter}},\ }\bibfield  {title} {\bibinfo {title} {Rescuing ecosystems from
  extinction cascades through compensatory perturbations},\ }\href
  {https://doi.org/10.1038/ncomms1163} {\bibfield  {journal} {\bibinfo
  {journal} {Nature Communications}\ }\textbf {\bibinfo {volume} {2}},\
  \bibinfo {pages} {170} (\bibinfo {year} {2011})}\BibitemShut {NoStop}%
\bibitem [{\citenamefont {Motter}(2004)}]{motter_cascade_2004}%
  \BibitemOpen
  \bibfield  {author} {\bibinfo {author} {\bibfnamefont {A.~E.}\ \bibnamefont
  {Motter}},\ }\bibfield  {title} {\bibinfo {title} {Cascade {Control} and
  {Defense} in {Complex} {Networks}},\ }\href
  {https://doi.org/10.1103/PhysRevLett.93.098701} {\bibfield  {journal}
  {\bibinfo  {journal} {Physical Review Letters}\ }\textbf {\bibinfo {volume}
  {93}},\ \bibinfo {pages} {098701} (\bibinfo {year} {2004})}\BibitemShut
  {NoStop}%
\bibitem [{\citenamefont {D'Souza}(2017)}]{dsouza_curtailing_2017}%
  \BibitemOpen
  \bibfield  {author} {\bibinfo {author} {\bibfnamefont {R.~M.}\ \bibnamefont
  {D'Souza}},\ }\bibfield  {title} {{\selectlanguage {english}\bibinfo {title}
  {Curtailing cascading failures}},\ }\href
  {https://doi.org/10.1126/science.aaq0474} {\bibfield  {journal} {\bibinfo
  {journal} {Science}\ }\textbf {\bibinfo {volume} {358}},\ \bibinfo {pages}
  {860} (\bibinfo {year} {2017})}\BibitemShut {NoStop}%
\bibitem [{\citenamefont {Lin}\ \emph {et~al.}(2018)\citenamefont {Lin},
  \citenamefont {Burghardt}, \citenamefont {Rohden}, \citenamefont {Noël},\
  and\ \citenamefont {D'Souza}}]{lin_self-organization_2018}%
  \BibitemOpen
  \bibfield  {author} {\bibinfo {author} {\bibfnamefont {Y.}~\bibnamefont
  {Lin}}, \bibinfo {author} {\bibfnamefont {K.}~\bibnamefont {Burghardt}},
  \bibinfo {author} {\bibfnamefont {M.}~\bibnamefont {Rohden}}, \bibinfo
  {author} {\bibfnamefont {P.-A.}\ \bibnamefont {Noël}},\ and\ \bibinfo
  {author} {\bibfnamefont {R.~M.}\ \bibnamefont {D'Souza}},\ }\bibfield
  {title} {{\selectlanguage {english}\bibinfo {title} {Self-organization of
  dragon king failures}},\ }\href {https://doi.org/10.1103/PhysRevE.98.022127}
  {\bibfield  {journal} {\bibinfo  {journal} {Physical Review E}\ }\textbf
  {\bibinfo {volume} {98}},\ \bibinfo {pages} {022127} (\bibinfo {year}
  {2018})}\BibitemShut {NoStop}%
\bibitem [{\citenamefont {Stouffer}\ and\ \citenamefont
  {Bascompte}(2011)}]{Stouffer2011}%
  \BibitemOpen
  \bibfield  {author} {\bibinfo {author} {\bibfnamefont {D.~B.}\ \bibnamefont
  {Stouffer}}\ and\ \bibinfo {author} {\bibfnamefont {J.}~\bibnamefont
  {Bascompte}},\ }\bibfield  {title} {\bibinfo {title} {Compartmentalization
  increases food-web persistence},\ }\href@noop {} {\bibfield  {journal}
  {\bibinfo  {journal} {Proceedings of the National Academy of Sciences}\
  }\textbf {\bibinfo {volume} {108}},\ \bibinfo {pages} {3648} (\bibinfo {year}
  {2011})}\BibitemShut {NoStop}%
\bibitem [{\citenamefont {Gilarranz}\ \emph {et~al.}(2017)\citenamefont
  {Gilarranz}, \citenamefont {Rayfield}, \citenamefont
  {Li{\~{n}}{\'{a}}n-Cembrano}, \citenamefont {Bascompte},\ and\ \citenamefont
  {Gonzalez}}]{gilarranz_effects_2017}%
  \BibitemOpen
  \bibfield  {author} {\bibinfo {author} {\bibfnamefont {L.~J.}\ \bibnamefont
  {Gilarranz}}, \bibinfo {author} {\bibfnamefont {B.}~\bibnamefont {Rayfield}},
  \bibinfo {author} {\bibfnamefont {G.}~\bibnamefont
  {Li{\~{n}}{\'{a}}n-Cembrano}}, \bibinfo {author} {\bibfnamefont
  {J.}~\bibnamefont {Bascompte}},\ and\ \bibinfo {author} {\bibfnamefont
  {A.}~\bibnamefont {Gonzalez}},\ }\bibfield  {title} {\bibinfo {title}
  {Effects of network modularity on the spread of perturbation impact in
  experimental metapopulations},\ }\href
  {https://doi.org/10.1126/science.aal4122} {\bibfield  {journal} {\bibinfo
  {journal} {Science}\ }\textbf {\bibinfo {volume} {357}},\ \bibinfo {pages}
  {199} (\bibinfo {year} {2017})}\BibitemShut {NoStop}%
\bibitem [{\citenamefont {Hens}\ \emph {et~al.}(2019)\citenamefont {Hens},
  \citenamefont {Harush}, \citenamefont {Haber}, \citenamefont {Cohen},\ and\
  \citenamefont {Barzel}}]{Hens2019}%
  \BibitemOpen
  \bibfield  {author} {\bibinfo {author} {\bibfnamefont {C.}~\bibnamefont
  {Hens}}, \bibinfo {author} {\bibfnamefont {U.}~\bibnamefont {Harush}},
  \bibinfo {author} {\bibfnamefont {S.}~\bibnamefont {Haber}}, \bibinfo
  {author} {\bibfnamefont {R.}~\bibnamefont {Cohen}},\ and\ \bibinfo {author}
  {\bibfnamefont {B.}~\bibnamefont {Barzel}},\ }\bibfield  {title} {\bibinfo
  {title} {Spatiotemporal signal propagation in complex networks},\ }\href@noop
  {} {\bibfield  {journal} {\bibinfo  {journal} {Nature Physics}\ }\textbf
  {\bibinfo {volume} {15}},\ \bibinfo {pages} {403} (\bibinfo {year}
  {2019})}\BibitemShut {NoStop}%
\bibitem [{\citenamefont {Brummitt}\ \emph {et~al.}(2012)\citenamefont
  {Brummitt}, \citenamefont {D'Souza},\ and\ \citenamefont
  {Leicht}}]{brummitt_suppressing_2012}%
  \BibitemOpen
  \bibfield  {author} {\bibinfo {author} {\bibfnamefont {C.~D.}\ \bibnamefont
  {Brummitt}}, \bibinfo {author} {\bibfnamefont {R.~M.}\ \bibnamefont
  {D'Souza}},\ and\ \bibinfo {author} {\bibfnamefont {E.~A.}\ \bibnamefont
  {Leicht}},\ }\bibfield  {title} {{\selectlanguage {english}\bibinfo {title}
  {Suppressing cascades of load in interdependent networks}},\ }\href
  {https://doi.org/10.1073/pnas.1110586109} {\bibfield  {journal} {\bibinfo
  {journal} {Proceedings of the National Academy of Sciences}\ }\textbf
  {\bibinfo {volume} {109}},\ \bibinfo {pages} {E680} (\bibinfo {year}
  {2012})}\BibitemShut {NoStop}%
\bibitem [{\citenamefont {Radicchi}\ \emph {et~al.}(2004)\citenamefont
  {Radicchi}, \citenamefont {Castellano}, \citenamefont {Cecconi},
  \citenamefont {Loreto},\ and\ \citenamefont
  {Parisi}}]{radicchi_defining_2004}%
  \BibitemOpen
  \bibfield  {author} {\bibinfo {author} {\bibfnamefont {F.}~\bibnamefont
  {Radicchi}}, \bibinfo {author} {\bibfnamefont {C.}~\bibnamefont
  {Castellano}}, \bibinfo {author} {\bibfnamefont {F.}~\bibnamefont {Cecconi}},
  \bibinfo {author} {\bibfnamefont {V.}~\bibnamefont {Loreto}},\ and\ \bibinfo
  {author} {\bibfnamefont {D.}~\bibnamefont {Parisi}},\ }\bibfield  {title}
  {\bibinfo {title} {Defining and identifying communities in networks},\ }\href
  {https://doi.org/10.1073/pnas.0400054101} {\bibfield  {journal} {\bibinfo
  {journal} {Proceedings of the National Academy of Sciences}\ }\textbf
  {\bibinfo {volume} {101}},\ \bibinfo {pages} {2658} (\bibinfo {year}
  {2004})}\BibitemShut {NoStop}%
\bibitem [{\citenamefont {Girvan}\ and\ \citenamefont
  {Newman}(2002)}]{girvan_community_2002}%
  \BibitemOpen
  \bibfield  {author} {\bibinfo {author} {\bibfnamefont {M.}~\bibnamefont
  {Girvan}}\ and\ \bibinfo {author} {\bibfnamefont {M.~E.~J.}\ \bibnamefont
  {Newman}},\ }\bibfield  {title} {\bibinfo {title} {Community structure in
  social and biological networks},\ }\href
  {https://doi.org/10.1073/pnas.122653799} {\bibfield  {journal} {\bibinfo
  {journal} {Proceedings of the National Academy of Sciences}\ }\textbf
  {\bibinfo {volume} {99}},\ \bibinfo {pages} {7821} (\bibinfo {year}
  {2002})}\BibitemShut {NoStop}%
\bibitem [{\citenamefont {Newman}(2012)}]{newman_communities_2012}%
  \BibitemOpen
  \bibfield  {author} {\bibinfo {author} {\bibfnamefont {M.~E.~J.}\
  \bibnamefont {Newman}},\ }\bibfield  {title} {\bibinfo {title} {Communities,
  modules and large-scale structure in networks},\ }\href
  {https://doi.org/10.1038/nphys2162} {\bibfield  {journal} {\bibinfo
  {journal} {Nature Physics}\ }\textbf {\bibinfo {volume} {8}},\ \bibinfo
  {pages} {25} (\bibinfo {year} {2012})}\BibitemShut {NoStop}%
\bibitem [{\citenamefont {H{\"o}rsch}\ \emph {et~al.}(2018)\citenamefont
  {H{\"o}rsch}, \citenamefont {Hofmann}, \citenamefont {Schlachtberger},\ and\
  \citenamefont {Brown}}]{horsch_2018}%
  \BibitemOpen
  \bibfield  {author} {\bibinfo {author} {\bibfnamefont {J.}~\bibnamefont
  {H{\"o}rsch}}, \bibinfo {author} {\bibfnamefont {F.}~\bibnamefont {Hofmann}},
  \bibinfo {author} {\bibfnamefont {D.}~\bibnamefont {Schlachtberger}},\ and\
  \bibinfo {author} {\bibfnamefont {T.}~\bibnamefont {Brown}},\ }\bibfield
  {title} {\bibinfo {title} {{PyPSA}-eur: An open optimisation model of the
  european transmission system},\ }\href
  {https://doi.org/10.1016/j.esr.2018.08.012} {\bibfield  {journal} {\bibinfo
  {journal} {Energy Strategy Reviews}\ }\textbf {\bibinfo {volume} {22}},\
  \bibinfo {pages} {207} (\bibinfo {year} {2018})}\BibitemShut {NoStop}%
\bibitem [{\citenamefont {Ronellenfitsch}\ \emph {et~al.}(2015)\citenamefont
  {Ronellenfitsch}, \citenamefont {Lasser}, \citenamefont {Daly},\ and\
  \citenamefont {Katifori}}]{ronellenfitsch_topological_2015}%
  \BibitemOpen
  \bibfield  {author} {\bibinfo {author} {\bibfnamefont {H.}~\bibnamefont
  {Ronellenfitsch}}, \bibinfo {author} {\bibfnamefont {J.}~\bibnamefont
  {Lasser}}, \bibinfo {author} {\bibfnamefont {D.~C.}\ \bibnamefont {Daly}},\
  and\ \bibinfo {author} {\bibfnamefont {E.}~\bibnamefont {Katifori}},\
  }\bibfield  {title} {{\selectlanguage {english}\bibinfo {title} {Topological
  {Phenotypes} {Constitute} a {New} {Dimension} in the {Phenotypic} {Space} of
  {Leaf} {Venation} {Networks}}},\ }\href
  {https://doi.org/10.1371/journal.pcbi.1004680} {\bibfield  {journal}
  {\bibinfo  {journal} {PLOS Computational Biology}\ }\textbf {\bibinfo
  {volume} {11}},\ \bibinfo {pages} {e1004680} (\bibinfo {year}
  {2015})}\BibitemShut {NoStop}%
\bibitem [{\citenamefont {Lasseter}(2002)}]{Lasseter2002}%
  \BibitemOpen
  \bibfield  {author} {\bibinfo {author} {\bibfnamefont {R.~H.}\ \bibnamefont
  {Lasseter}},\ }\bibfield  {title} {\bibinfo {title} {Microgrids},\ }in\
  \href@noop {} {\emph {\bibinfo {booktitle} {2002 IEEE Power Engineering
  Society Winter Meeting. Conference Proceedings (Cat. No. 02CH37309)}}},\
  Vol.~\bibinfo {volume} {1}\ (\bibinfo {organization} {IEEE},\ \bibinfo {year}
  {2002})\ pp.\ \bibinfo {pages} {305--308}\BibitemShut {NoStop}%
\bibitem [{\citenamefont {Mureddu}\ \emph {et~al.}(2016)\citenamefont
  {Mureddu}, \citenamefont {Caldarelli}, \citenamefont {Damiano}, \citenamefont
  {Scala},\ and\ \citenamefont {Meyer-Ortmanns}}]{mureddu2016islanding}%
  \BibitemOpen
  \bibfield  {author} {\bibinfo {author} {\bibfnamefont {M.}~\bibnamefont
  {Mureddu}}, \bibinfo {author} {\bibfnamefont {G.}~\bibnamefont {Caldarelli}},
  \bibinfo {author} {\bibfnamefont {A.}~\bibnamefont {Damiano}}, \bibinfo
  {author} {\bibfnamefont {A.}~\bibnamefont {Scala}},\ and\ \bibinfo {author}
  {\bibfnamefont {H.}~\bibnamefont {Meyer-Ortmanns}},\ }\bibfield  {title}
  {\bibinfo {title} {Islanding the power grid on the transmission level: less
  connections for more security},\ }\href@noop {} {\bibfield  {journal}
  {\bibinfo  {journal} {Scientific Reports}\ }\textbf {\bibinfo {volume} {6}},\
  \bibinfo {pages} {34797} (\bibinfo {year} {2016})}\BibitemShut {NoStop}%
\bibitem [{\citenamefont {Schlachtberger}\ \emph {et~al.}(2017)\citenamefont
  {Schlachtberger}, \citenamefont {Brown}, \citenamefont {Schramm},\ and\
  \citenamefont {Greiner}}]{Schlachtberger2017}%
  \BibitemOpen
  \bibfield  {author} {\bibinfo {author} {\bibfnamefont {D.~P.}\ \bibnamefont
  {Schlachtberger}}, \bibinfo {author} {\bibfnamefont {T.}~\bibnamefont
  {Brown}}, \bibinfo {author} {\bibfnamefont {S.}~\bibnamefont {Schramm}},\
  and\ \bibinfo {author} {\bibfnamefont {M.}~\bibnamefont {Greiner}},\
  }\bibfield  {title} {\bibinfo {title} {The benefits of cooperation in a
  highly renewable european electricity network},\ }\href@noop {} {\bibfield
  {journal} {\bibinfo  {journal} {Energy}\ }\textbf {\bibinfo {volume} {134}},\
  \bibinfo {pages} {469} (\bibinfo {year} {2017})}\BibitemShut {NoStop}%
\bibitem [{\citenamefont {Tröndle}\ \emph {et~al.}(2020)\citenamefont
  {Tröndle}, \citenamefont {Lilliestam}, \citenamefont {Marelli},\ and\
  \citenamefont {Pfenninger}}]{TRONDLE2020}%
  \BibitemOpen
  \bibfield  {author} {\bibinfo {author} {\bibfnamefont {T.}~\bibnamefont
  {Tröndle}}, \bibinfo {author} {\bibfnamefont {J.}~\bibnamefont
  {Lilliestam}}, \bibinfo {author} {\bibfnamefont {S.}~\bibnamefont
  {Marelli}},\ and\ \bibinfo {author} {\bibfnamefont {S.}~\bibnamefont
  {Pfenninger}},\ }\bibfield  {title} {\bibinfo {title} {Trade-offs between
  geographic scale, cost, and infrastructure requirements for fully renewable
  electricity in europe},\ }\href@noop {} {\bibfield  {journal} {\bibinfo
  {journal} {Joule}\ } (\bibinfo {year} {2020})}\BibitemShut {NoStop}%
\bibitem [{\citenamefont {Gavrilchenko}\ and\ \citenamefont
  {Katifori}(2019)}]{gavrilchenko_resilience_2019}%
  \BibitemOpen
  \bibfield  {author} {\bibinfo {author} {\bibfnamefont {T.}~\bibnamefont
  {Gavrilchenko}}\ and\ \bibinfo {author} {\bibfnamefont {E.}~\bibnamefont
  {Katifori}},\ }\bibfield  {title} {\bibinfo {title} {Resilience in
  hierarchical fluid flow networks},\ }\href
  {https://doi.org/10.1103/PhysRevE.99.012321} {\bibfield  {journal} {\bibinfo
  {journal} {Physical Review E}\ }\textbf {\bibinfo {volume} {99}},\ \bibinfo
  {pages} {012321} (\bibinfo {year} {2019})}\BibitemShut {NoStop}%
\bibitem [{Note1()}]{Note1}%
  \BibitemOpen
  \bibinfo {note} {See Supplemental Material for a proof of the main theorem on
  network isolators, extended Figures on their applicability, a discussion of
  their relationship to network controllability and detailed methods, which
  include Refs.~\cite
  {newman2010,bollobas1998,Dorfler2018,Wood14,motter_ml_2002,crucitti_model_2004,yuan_exact_2013,van_mieghem_2017,gao_target_2014,liu_controllability_2011,Norman97,Rohden2012,nishikawa2015,manik_supply_2014,rodrigues_kuramoto_2016,wurbs2001,reichold_vascular_2009}}\BibitemShut
  {NoStop}%
\bibitem [{\citenamefont {Purchala}\ \emph {et~al.}(2005)\citenamefont
  {Purchala}, \citenamefont {Meeus}, \citenamefont {Dommelen},\ and\
  \citenamefont {Belmans}}]{Purc05}%
  \BibitemOpen
  \bibfield  {author} {\bibinfo {author} {\bibfnamefont {K.}~\bibnamefont
  {Purchala}}, \bibinfo {author} {\bibfnamefont {L.}~\bibnamefont {Meeus}},
  \bibinfo {author} {\bibfnamefont {D.~V.}\ \bibnamefont {Dommelen}},\ and\
  \bibinfo {author} {\bibfnamefont {R.}~\bibnamefont {Belmans}},\ }\bibfield
  {title} {\bibinfo {title} {Usefulness of dc power flow for active power flow
  analysis},\ }in\ \href {https://doi.org/10.1109/PES.2005.1489581} {\emph
  {\bibinfo {booktitle} {IEEE Power Engineering Society General Meeting}}}\
  (\bibinfo {year} {2005})\ pp.\ \bibinfo {pages} {454--459 Vol. 1}\BibitemShut
  {NoStop}%
\bibitem [{\citenamefont {Katifori}\ \emph {et~al.}(2010)\citenamefont
  {Katifori}, \citenamefont {Sz{\"o}ll{\H{o}}si},\ and\ \citenamefont
  {Magnasco}}]{Kati10}%
  \BibitemOpen
  \bibfield  {author} {\bibinfo {author} {\bibfnamefont {E.}~\bibnamefont
  {Katifori}}, \bibinfo {author} {\bibfnamefont {G.~J.}\ \bibnamefont
  {Sz{\"o}ll{\H{o}}si}},\ and\ \bibinfo {author} {\bibfnamefont {M.~O.}\
  \bibnamefont {Magnasco}},\ }\bibfield  {title} {\bibinfo {title} {Damage and
  fluctuations induce loops in optimal transport networks},\ }\href
  {https://doi.org/10.1103/PhysRevLett.104.048704} {\bibfield  {journal}
  {\bibinfo  {journal} {Physical Review Letters}\ }\textbf {\bibinfo {volume}
  {104}},\ \bibinfo {pages} {048704} (\bibinfo {year} {2010})}\BibitemShut
  {NoStop}%
\bibitem [{\citenamefont {Coomes}\ \emph {et~al.}(2008)\citenamefont {Coomes},
  \citenamefont {Heathcote}, \citenamefont {Godfrey}, \citenamefont
  {Shepherd},\ and\ \citenamefont {Sack}}]{coomes_scaling_2008}%
  \BibitemOpen
  \bibfield  {author} {\bibinfo {author} {\bibfnamefont {D.~A.}\ \bibnamefont
  {Coomes}}, \bibinfo {author} {\bibfnamefont {S.}~\bibnamefont {Heathcote}},
  \bibinfo {author} {\bibfnamefont {E.~R.}\ \bibnamefont {Godfrey}}, \bibinfo
  {author} {\bibfnamefont {J.~J.}\ \bibnamefont {Shepherd}},\ and\ \bibinfo
  {author} {\bibfnamefont {L.}~\bibnamefont {Sack}},\ }\bibfield  {title}
  {\bibinfo {title} {Scaling of xylem vessels and veins within the leaves of
  oak species},\ }\href {https://doi.org/10.1098/rsbl.2008.0094} {\bibfield
  {journal} {\bibinfo  {journal} {Biology Letters}\ }\textbf {\bibinfo {volume}
  {4}},\ \bibinfo {pages} {302} (\bibinfo {year} {2008})}\BibitemShut {NoStop}%
\bibitem [{\citenamefont {Manik}\ \emph {et~al.}(2017)\citenamefont {Manik},
  \citenamefont {Rohden}, \citenamefont {Ronellenfitsch}, \citenamefont
  {Zhang}, \citenamefont {Hallerberg}, \citenamefont {Witthaut},\ and\
  \citenamefont {Timme}}]{Manik2017}%
  \BibitemOpen
  \bibfield  {author} {\bibinfo {author} {\bibfnamefont {D.}~\bibnamefont
  {Manik}}, \bibinfo {author} {\bibfnamefont {M.}~\bibnamefont {Rohden}},
  \bibinfo {author} {\bibfnamefont {H.}~\bibnamefont {Ronellenfitsch}},
  \bibinfo {author} {\bibfnamefont {X.}~\bibnamefont {Zhang}}, \bibinfo
  {author} {\bibfnamefont {S.}~\bibnamefont {Hallerberg}}, \bibinfo {author}
  {\bibfnamefont {D.}~\bibnamefont {Witthaut}},\ and\ \bibinfo {author}
  {\bibfnamefont {M.}~\bibnamefont {Timme}},\ }\bibfield  {title} {\bibinfo
  {title} {Network susceptibilities: Theory and applications},\ }\href
  {https://doi.org/10.1103/PhysRevE.95.012319} {\bibfield  {journal} {\bibinfo
  {journal} {Physical Review E}\ }\textbf {\bibinfo {volume} {95}},\ \bibinfo
  {pages} {012319} (\bibinfo {year} {2017})}\BibitemShut {NoStop}%
\bibitem [{\citenamefont {Strake}\ \emph {et~al.}(2019)\citenamefont {Strake},
  \citenamefont {Kaiser}, \citenamefont {Basiri}, \citenamefont
  {Ronellenfitsch},\ and\ \citenamefont {Witthaut}}]{strake2018}%
  \BibitemOpen
  \bibfield  {author} {\bibinfo {author} {\bibfnamefont {J.}~\bibnamefont
  {Strake}}, \bibinfo {author} {\bibfnamefont {F.}~\bibnamefont {Kaiser}},
  \bibinfo {author} {\bibfnamefont {F.}~\bibnamefont {Basiri}}, \bibinfo
  {author} {\bibfnamefont {H.}~\bibnamefont {Ronellenfitsch}},\ and\ \bibinfo
  {author} {\bibfnamefont {D.}~\bibnamefont {Witthaut}},\ }\bibfield  {title}
  {\bibinfo {title} {Non-local impact of link failures in linear flow
  networks},\ }\href {https://doi.org/10.1088/1367-2630/ab13ba} {\bibfield
  {journal} {\bibinfo  {journal} {New Journal of Physics}\ }\textbf {\bibinfo
  {volume} {21}},\ \bibinfo {pages} {053009} (\bibinfo {year}
  {2019})}\BibitemShut {NoStop}%
\bibitem [{\citenamefont {Fortunato}(2010)}]{fortunato_community_2010}%
  \BibitemOpen
  \bibfield  {author} {\bibinfo {author} {\bibfnamefont {S.}~\bibnamefont
  {Fortunato}},\ }\bibfield  {title} {\bibinfo {title} {Community detection in
  graphs},\ }\href@noop {} {\bibfield  {journal} {\bibinfo  {journal} {Physics
  Reports}\ }\textbf {\bibinfo {volume} {486}},\ \bibinfo {pages} {75}
  (\bibinfo {year} {2010})}\BibitemShut {NoStop}%
\bibitem [{\citenamefont {Kaiser}\ \emph {et~al.}(2020)\citenamefont {Kaiser},
  \citenamefont {Strake},\ and\ \citenamefont
  {Witthaut}}]{kaiser_collective_2020}%
  \BibitemOpen
  \bibfield  {author} {\bibinfo {author} {\bibfnamefont {F.}~\bibnamefont
  {Kaiser}}, \bibinfo {author} {\bibfnamefont {J.}~\bibnamefont {Strake}},\
  and\ \bibinfo {author} {\bibfnamefont {D.}~\bibnamefont {Witthaut}},\
  }\bibfield  {title} {{\selectlanguage {english}\bibinfo {title} {Collective
  effects of link failures in linear flow networks}},\ }\href
  {https://doi.org/10.1088/1367-2630/ab6793} {\bibfield  {journal} {\bibinfo
  {journal} {New Journal of Physics}\ }\textbf {\bibinfo {volume} {22}},\
  \bibinfo {pages} {013053} (\bibinfo {year} {2020})}\BibitemShut {NoStop}%
\bibitem [{\citenamefont {Kettemann}(2015)}]{Kett15}%
  \BibitemOpen
  \bibfield  {author} {\bibinfo {author} {\bibfnamefont {S.}~\bibnamefont
  {Kettemann}},\ }\bibfield  {title} {\bibinfo {title} {Delocalization of
  disturbances and the stability of ac electricity grids},\ }\href
  {https://doi.org/10.1103/PhysRevE.94.062311} {\bibfield  {journal} {\bibinfo
  {journal} {Physical Review E}\ }\textbf {\bibinfo {volume} {94}},\ \bibinfo
  {pages} {062311} (\bibinfo {year} {2015})}\BibitemShut {NoStop}%
\bibitem [{\citenamefont {Erd\H{o}s}(1960)}]{Erdos1960}%
  \BibitemOpen
  \bibfield  {author} {\bibinfo {author} {\bibfnamefont {P.}~\bibnamefont
  {Erd\H{o}s}},\ }\bibfield  {title} {\bibinfo {title} {On the evolution of
  random graphs},\ }\href@noop {} {\bibfield  {journal} {\bibinfo  {journal}
  {Publications of the Mathematical Institute of the Hungarian Academy of
  Sciences}\ }\textbf {\bibinfo {volume} {5}},\ \bibinfo {pages} {17} (\bibinfo
  {year} {1960})}\BibitemShut {NoStop}%
\bibitem [{\citenamefont {Pecora}\ \emph {et~al.}(2014)\citenamefont {Pecora},
  \citenamefont {Sorrentino}, \citenamefont {Hagerstrom}, \citenamefont
  {Murphy},\ and\ \citenamefont {Roy}}]{pecora_cluster_2014}%
  \BibitemOpen
  \bibfield  {author} {\bibinfo {author} {\bibfnamefont {L.~M.}\ \bibnamefont
  {Pecora}}, \bibinfo {author} {\bibfnamefont {F.}~\bibnamefont {Sorrentino}},
  \bibinfo {author} {\bibfnamefont {A.~M.}\ \bibnamefont {Hagerstrom}},
  \bibinfo {author} {\bibfnamefont {T.~E.}\ \bibnamefont {Murphy}},\ and\
  \bibinfo {author} {\bibfnamefont {R.}~\bibnamefont {Roy}},\ }\bibfield
  {title} {{\selectlanguage {english}\bibinfo {title} {Cluster synchronization
  and isolated desynchronization in complex networks with symmetries}},\ }\href
  {https://doi.org/10.1038/ncomms5079} {\bibfield  {journal} {\bibinfo
  {journal} {Nature Communications}\ }\textbf {\bibinfo {volume} {5}},\
  \bibinfo {pages} {1} (\bibinfo {year} {2014})}\BibitemShut {NoStop}%
\bibitem [{\citenamefont {Sorrentino}\ \emph {et~al.}(2016)\citenamefont
  {Sorrentino}, \citenamefont {Pecora}, \citenamefont {Hagerstrom},
  \citenamefont {Murphy},\ and\ \citenamefont
  {Roy}}]{sorrentino_complete_2016}%
  \BibitemOpen
  \bibfield  {author} {\bibinfo {author} {\bibfnamefont {F.}~\bibnamefont
  {Sorrentino}}, \bibinfo {author} {\bibfnamefont {L.~M.}\ \bibnamefont
  {Pecora}}, \bibinfo {author} {\bibfnamefont {A.~M.}\ \bibnamefont
  {Hagerstrom}}, \bibinfo {author} {\bibfnamefont {T.~E.}\ \bibnamefont
  {Murphy}},\ and\ \bibinfo {author} {\bibfnamefont {R.}~\bibnamefont {Roy}},\
  }\bibfield  {title} {{\selectlanguage {english}\bibinfo {title} {Complete
  characterization of the stability of cluster synchronization in complex
  dynamical networks}},\ }\href {https://doi.org/10.1126/sciadv.1501737}
  {\bibfield  {journal} {\bibinfo  {journal} {Science Advances}\ }\textbf
  {\bibinfo {volume} {2}},\ \bibinfo {pages} {e1501737} (\bibinfo {year}
  {2016})}\BibitemShut {NoStop}%
\bibitem [{\citenamefont {Nicosia}\ \emph {et~al.}(2013)\citenamefont
  {Nicosia}, \citenamefont {Valencia}, \citenamefont {Chavez}, \citenamefont
  {D\'iaz-Guilera},\ and\ \citenamefont {Latora}}]{nicosia_remote_2013}%
  \BibitemOpen
  \bibfield  {author} {\bibinfo {author} {\bibfnamefont {V.}~\bibnamefont
  {Nicosia}}, \bibinfo {author} {\bibfnamefont {M.}~\bibnamefont {Valencia}},
  \bibinfo {author} {\bibfnamefont {M.}~\bibnamefont {Chavez}}, \bibinfo
  {author} {\bibfnamefont {A.}~\bibnamefont {D\'iaz-Guilera}},\ and\ \bibinfo
  {author} {\bibfnamefont {V.}~\bibnamefont {Latora}},\ }\bibfield  {title}
  {{\selectlanguage {english}\bibinfo {title} {Remote {Synchronization}
  {Reveals} {Network} {Symmetries} and {Functional} {Modules}}},\ }\href
  {https://doi.org/10.1103/PhysRevLett.110.174102} {\bibfield  {journal}
  {\bibinfo  {journal} {Physical Review Letters}\ }\textbf {\bibinfo {volume}
  {110}},\ \bibinfo {pages} {174102} (\bibinfo {year} {2013})}\BibitemShut
  {NoStop}%
\bibitem [{\citenamefont {Tropp}\ and\ \citenamefont
  {Wright}(2010)}]{tropp_computational_2010}%
  \BibitemOpen
  \bibfield  {author} {\bibinfo {author} {\bibfnamefont {J.~A.}\ \bibnamefont
  {Tropp}}\ and\ \bibinfo {author} {\bibfnamefont {S.~J.}\ \bibnamefont
  {Wright}},\ }\bibfield  {title} {\bibinfo {title} {Computational methods for
  sparse solution of linear inverse problems},\ }\href
  {https://doi.org/10.1109/JPROC.2010.2044010} {\bibfield  {journal} {\bibinfo
  {journal} {Proceedings of the {IEEE}}\ }\textbf {\bibinfo {volume} {98}},\
  \bibinfo {pages} {948} (\bibinfo {year} {2010})}\BibitemShut {NoStop}%
\bibitem [{\citenamefont {Sch{\"a}fer}\ \emph {et~al.}(2018)\citenamefont
  {Sch{\"a}fer}, \citenamefont {Witthaut}, \citenamefont {Timme},\ and\
  \citenamefont {Latora}}]{Schafer2018}%
  \BibitemOpen
  \bibfield  {author} {\bibinfo {author} {\bibfnamefont {B.}~\bibnamefont
  {Sch{\"a}fer}}, \bibinfo {author} {\bibfnamefont {D.}~\bibnamefont
  {Witthaut}}, \bibinfo {author} {\bibfnamefont {M.}~\bibnamefont {Timme}},\
  and\ \bibinfo {author} {\bibfnamefont {V.}~\bibnamefont {Latora}},\
  }\bibfield  {title} {\bibinfo {title} {Dynamically induced cascading failures
  in power grids},\ }\href@noop {} {\bibfield  {journal} {\bibinfo  {journal}
  {Nature Communications}\ }\textbf {\bibinfo {volume} {9}},\ \bibinfo {pages}
  {1975} (\bibinfo {year} {2018})}\BibitemShut {NoStop}%
\bibitem [{\citenamefont {Newman}(2010)}]{newman2010}%
  \BibitemOpen
  \bibfield  {author} {\bibinfo {author} {\bibfnamefont {M.~E.~J.}\
  \bibnamefont {Newman}},\ }\href@noop {} {\emph {\bibinfo {title} {Networks:
  An introduction}}}\ (\bibinfo  {publisher} {Oxford University Press},\
  \bibinfo {year} {2010})\BibitemShut {NoStop}%
\bibitem [{\citenamefont {Bollob{\'a}s}(1998)}]{bollobas1998}%
  \BibitemOpen
  \bibfield  {author} {\bibinfo {author} {\bibfnamefont {B.}~\bibnamefont
  {Bollob{\'a}s}},\ }\href@noop {} {\emph {\bibinfo {title} {Modern graph
  theory}}},\ \bibinfo {series} {Graduate texts in mathematics}\ No.\ \bibinfo
  {number} {184}\ (\bibinfo  {publisher} {Springer},\ \bibinfo {year}
  {1998})\BibitemShut {NoStop}%
\bibitem [{\citenamefont {D{\"o}rfler}\ \emph {et~al.}(2018)\citenamefont
  {D{\"o}rfler}, \citenamefont {Simpson-Porco},\ and\ \citenamefont
  {Bullo}}]{Dorfler2018}%
  \BibitemOpen
  \bibfield  {author} {\bibinfo {author} {\bibfnamefont {F.}~\bibnamefont
  {D{\"o}rfler}}, \bibinfo {author} {\bibfnamefont {J.~W.}\ \bibnamefont
  {Simpson-Porco}},\ and\ \bibinfo {author} {\bibfnamefont {F.}~\bibnamefont
  {Bullo}},\ }\bibfield  {title} {\bibinfo {title} {Electrical networks and
  algebraic graph theory: Models, properties, and applications},\ }\href@noop
  {} {\bibfield  {journal} {\bibinfo  {journal} {Proceedings of the IEEE}\
  }\textbf {\bibinfo {volume} {106}},\ \bibinfo {pages} {977} (\bibinfo {year}
  {2018})}\BibitemShut {NoStop}%
\bibitem [{\citenamefont {Wood}\ \emph {et~al.}(2014)\citenamefont {Wood},
  \citenamefont {Wollenberg},\ and\ \citenamefont {Shebl\'e}}]{Wood14}%
  \BibitemOpen
  \bibfield  {author} {\bibinfo {author} {\bibfnamefont {A.~J.}\ \bibnamefont
  {Wood}}, \bibinfo {author} {\bibfnamefont {B.~F.}\ \bibnamefont
  {Wollenberg}},\ and\ \bibinfo {author} {\bibfnamefont {G.~B.}\ \bibnamefont
  {Shebl\'e}},\ }\href@noop {} {\emph {\bibinfo {title} {Power Generation,
  Operation and Control}}}\ (\bibinfo  {publisher} {John Wiley \& Sons},\
  \bibinfo {address} {New York},\ \bibinfo {year} {2014})\BibitemShut {NoStop}%
\bibitem [{\citenamefont {Motter}\ and\ \citenamefont
  {Lai}(2002)}]{motter_ml_2002}%
  \BibitemOpen
  \bibfield  {author} {\bibinfo {author} {\bibfnamefont {A.~E.}\ \bibnamefont
  {Motter}}\ and\ \bibinfo {author} {\bibfnamefont {Y.-C.}\ \bibnamefont
  {Lai}},\ }\bibfield  {title} {\bibinfo {title} {Cascade-based attacks on
  complex networks},\ }\href {https://doi.org/10.1103/PhysRevE.66.065102}
  {\bibfield  {journal} {\bibinfo  {journal} {Physical Review E}\ }\textbf
  {\bibinfo {volume} {66}},\ \bibinfo {pages} {065102} (\bibinfo {year}
  {2002})}\BibitemShut {NoStop}%
\bibitem [{\citenamefont {Crucitti}\ \emph {et~al.}(2004)\citenamefont
  {Crucitti}, \citenamefont {Latora},\ and\ \citenamefont
  {Marchiori}}]{crucitti_model_2004}%
  \BibitemOpen
  \bibfield  {author} {\bibinfo {author} {\bibfnamefont {P.}~\bibnamefont
  {Crucitti}}, \bibinfo {author} {\bibfnamefont {V.}~\bibnamefont {Latora}},\
  and\ \bibinfo {author} {\bibfnamefont {M.}~\bibnamefont {Marchiori}},\
  }\bibfield  {title} {\bibinfo {title} {Model for cascading failures in
  complex networks},\ }\href {https://doi.org/10.1103/PhysRevE.69.045104}
  {\bibfield  {journal} {\bibinfo  {journal} {Physical Review E}\ }\textbf
  {\bibinfo {volume} {69}},\ \bibinfo {pages} {045104} (\bibinfo {year}
  {2004})}\BibitemShut {NoStop}%
\bibitem [{\citenamefont {Yuan}\ \emph {et~al.}(2013)\citenamefont {Yuan},
  \citenamefont {Zhao}, \citenamefont {Di}, \citenamefont {Wang},\ and\
  \citenamefont {Lai}}]{yuan_exact_2013}%
  \BibitemOpen
  \bibfield  {author} {\bibinfo {author} {\bibfnamefont {Z.}~\bibnamefont
  {Yuan}}, \bibinfo {author} {\bibfnamefont {C.}~\bibnamefont {Zhao}}, \bibinfo
  {author} {\bibfnamefont {Z.}~\bibnamefont {Di}}, \bibinfo {author}
  {\bibfnamefont {W.-X.}\ \bibnamefont {Wang}},\ and\ \bibinfo {author}
  {\bibfnamefont {Y.-C.}\ \bibnamefont {Lai}},\ }\bibfield  {title}
  {{\selectlanguage {english}\bibinfo {title} {Exact controllability of complex
  networks}},\ }\href {https://doi.org/10.1038/ncomms3447} {\bibfield
  {journal} {\bibinfo  {journal} {Nature Communications}\ }\textbf {\bibinfo
  {volume} {4}},\ \bibinfo {pages} {2447} (\bibinfo {year} {2013})}\BibitemShut
  {NoStop}%
\bibitem [{\citenamefont {Van~Mieghem}\ \emph {et~al.}(2017)\citenamefont
  {Van~Mieghem}, \citenamefont {Devriendt},\ and\ \citenamefont
  {Cetinay}}]{van_mieghem_2017}%
  \BibitemOpen
  \bibfield  {author} {\bibinfo {author} {\bibfnamefont {P.}~\bibnamefont
  {Van~Mieghem}}, \bibinfo {author} {\bibfnamefont {K.}~\bibnamefont
  {Devriendt}},\ and\ \bibinfo {author} {\bibfnamefont {H.}~\bibnamefont
  {Cetinay}},\ }\bibfield  {title} {\bibinfo {title} {Pseudoinverse of the
  laplacian and best spreader node in a network},\ }\href
  {https://doi.org/10.1103/PhysRevE.96.032311} {\bibfield  {journal} {\bibinfo
  {journal} {Physical Review E}\ }\textbf {\bibinfo {volume} {96}},\ \bibinfo
  {pages} {032311} (\bibinfo {year} {2017})}\BibitemShut {NoStop}%
\bibitem [{\citenamefont {Gao}\ \emph {et~al.}(2014)\citenamefont {Gao},
  \citenamefont {Liu}, \citenamefont {D'Souza},\ and\ \citenamefont
  {Barabási}}]{gao_target_2014}%
  \BibitemOpen
  \bibfield  {author} {\bibinfo {author} {\bibfnamefont {J.}~\bibnamefont
  {Gao}}, \bibinfo {author} {\bibfnamefont {Y.-Y.}\ \bibnamefont {Liu}},
  \bibinfo {author} {\bibfnamefont {R.~M.}\ \bibnamefont {D'Souza}},\ and\
  \bibinfo {author} {\bibfnamefont {A.-L.}\ \bibnamefont {Barabási}},\
  }\bibfield  {title} {{\selectlanguage {english}\bibinfo {title} {Target
  control of complex networks}},\ }\href {https://doi.org/10.1038/ncomms6415}
  {\bibfield  {journal} {\bibinfo  {journal} {Nature Communications}\ }\textbf
  {\bibinfo {volume} {5}},\ \bibinfo {pages} {5415} (\bibinfo {year}
  {2014})}\BibitemShut {NoStop}%
\bibitem [{\citenamefont {Liu}\ \emph {et~al.}(2011)\citenamefont {Liu},
  \citenamefont {Slotine},\ and\ \citenamefont
  {Barabási}}]{liu_controllability_2011}%
  \BibitemOpen
  \bibfield  {author} {\bibinfo {author} {\bibfnamefont {Y.-Y.}\ \bibnamefont
  {Liu}}, \bibinfo {author} {\bibfnamefont {J.-J.}\ \bibnamefont {Slotine}},\
  and\ \bibinfo {author} {\bibfnamefont {A.-L.}\ \bibnamefont {Barabási}},\
  }\bibfield  {title} {{\selectlanguage {english}\bibinfo {title}
  {Controllability of complex networks}},\ }\href
  {https://doi.org/10.1038/nature10011} {\bibfield  {journal} {\bibinfo
  {journal} {Nature}\ }\textbf {\bibinfo {volume} {473}},\ \bibinfo {pages}
  {167} (\bibinfo {year} {2011})}\BibitemShut {NoStop}%
\bibitem [{\citenamefont {Biggs}(1997)}]{Norman97}%
  \BibitemOpen
  \bibfield  {author} {\bibinfo {author} {\bibfnamefont {N.}~\bibnamefont
  {Biggs}},\ }\bibfield  {title} {\bibinfo {title} {Algebraic potential theory
  on graphs},\ }\href@noop {} {\bibfield  {journal} {\bibinfo  {journal}
  {Bulletin of the London Mathematical Society}\ }\textbf {\bibinfo {volume}
  {29}},\ \bibinfo {pages} {641} (\bibinfo {year} {1997})}\BibitemShut
  {NoStop}%
\bibitem [{\citenamefont {Rohden}\ \emph {et~al.}(2012)\citenamefont {Rohden},
  \citenamefont {Sorge}, \citenamefont {Timme},\ and\ \citenamefont
  {Witthaut}}]{Rohden2012}%
  \BibitemOpen
  \bibfield  {author} {\bibinfo {author} {\bibfnamefont {M.}~\bibnamefont
  {Rohden}}, \bibinfo {author} {\bibfnamefont {A.}~\bibnamefont {Sorge}},
  \bibinfo {author} {\bibfnamefont {M.}~\bibnamefont {Timme}},\ and\ \bibinfo
  {author} {\bibfnamefont {D.}~\bibnamefont {Witthaut}},\ }\bibfield  {title}
  {\bibinfo {title} {Self-organized synchronization in decentralized power
  grids},\ }\href {https://doi.org/10.1103/PhysRevLett.109.064101} {\bibfield
  {journal} {\bibinfo  {journal} {Phys. Rev. Lett.}\ }\textbf {\bibinfo
  {volume} {109}},\ \bibinfo {pages} {064101} (\bibinfo {year}
  {2012})}\BibitemShut {NoStop}%
\bibitem [{\citenamefont {Nishikawa}\ and\ \citenamefont
  {Motter}(2015)}]{nishikawa2015}%
  \BibitemOpen
  \bibfield  {author} {\bibinfo {author} {\bibfnamefont {T.}~\bibnamefont
  {Nishikawa}}\ and\ \bibinfo {author} {\bibfnamefont {A.~E.}\ \bibnamefont
  {Motter}},\ }\bibfield  {title} {\bibinfo {title} {Comparative analysis of
  existing models for power-grid synchronization},\ }\href
  {https://doi.org/10.1088/1367-2630/17/1/015012} {\bibfield  {journal}
  {\bibinfo  {journal} {New Journal of Physics}\ }\textbf {\bibinfo {volume}
  {17}},\ \bibinfo {pages} {015012} (\bibinfo {year} {2015})}\BibitemShut
  {NoStop}%
\bibitem [{\citenamefont {Manik}\ \emph {et~al.}(2014)\citenamefont {Manik},
  \citenamefont {Witthaut}, \citenamefont {Sch{\"a}fer}, \citenamefont
  {Matthiae}, \citenamefont {Sorge}, \citenamefont {Rohden}, \citenamefont
  {Katifori},\ and\ \citenamefont {Timme}}]{manik_supply_2014}%
  \BibitemOpen
  \bibfield  {author} {\bibinfo {author} {\bibfnamefont {D.}~\bibnamefont
  {Manik}}, \bibinfo {author} {\bibfnamefont {D.}~\bibnamefont {Witthaut}},
  \bibinfo {author} {\bibfnamefont {B.}~\bibnamefont {Sch{\"a}fer}}, \bibinfo
  {author} {\bibfnamefont {M.}~\bibnamefont {Matthiae}}, \bibinfo {author}
  {\bibfnamefont {A.}~\bibnamefont {Sorge}}, \bibinfo {author} {\bibfnamefont
  {M.}~\bibnamefont {Rohden}}, \bibinfo {author} {\bibfnamefont
  {E.}~\bibnamefont {Katifori}},\ and\ \bibinfo {author} {\bibfnamefont
  {M.}~\bibnamefont {Timme}},\ }\bibfield  {title} {{\selectlanguage
  {english}\bibinfo {title} {Supply networks: {Instabilities} without
  overload}},\ }\href {https://doi.org/10.1140/epjst/e2014-02274-y} {\bibfield
  {journal} {\bibinfo  {journal} {The European Physical Journal Special
  Topics}\ }\textbf {\bibinfo {volume} {223}},\ \bibinfo {pages} {2527}
  (\bibinfo {year} {2014})}\BibitemShut {NoStop}%
\bibitem [{\citenamefont {Rodrigues}\ \emph {et~al.}(2016)\citenamefont
  {Rodrigues}, \citenamefont {Peron}, \citenamefont {Ji},\ and\ \citenamefont
  {Kurths}}]{rodrigues_kuramoto_2016}%
  \BibitemOpen
  \bibfield  {author} {\bibinfo {author} {\bibfnamefont {F.~A.}\ \bibnamefont
  {Rodrigues}}, \bibinfo {author} {\bibfnamefont {T.~K.~D.}\ \bibnamefont
  {Peron}}, \bibinfo {author} {\bibfnamefont {P.}~\bibnamefont {Ji}},\ and\
  \bibinfo {author} {\bibfnamefont {J.}~\bibnamefont {Kurths}},\ }\bibfield
  {title} {\bibinfo {title} {The kuramoto model in complex networks},\ }\href
  {https://doi.org/10.1016/j.physrep.2015.10.008} {\bibfield  {journal}
  {\bibinfo  {journal} {Physics Reports}\ }\textbf {\bibinfo {volume} {610}},\
  \bibinfo {pages} {1} (\bibinfo {year} {2016})}\BibitemShut {NoStop}%
\bibitem [{\citenamefont {Wurbs}\ and\ \citenamefont
  {James}(2001)}]{wurbs2001}%
  \BibitemOpen
  \bibfield  {author} {\bibinfo {author} {\bibfnamefont {R.~A.}\ \bibnamefont
  {Wurbs}}\ and\ \bibinfo {author} {\bibfnamefont {W.~P.}\ \bibnamefont
  {James}},\ }\href@noop {} {\emph {\bibinfo {title} {Water Resources
  Engineering}}}\ (\bibinfo  {publisher} {Pearson},\ \bibinfo {year}
  {2001})\BibitemShut {NoStop}%
\bibitem [{\citenamefont {Reichold}\ \emph {et~al.}(2009)\citenamefont
  {Reichold}, \citenamefont {Stampanoni}, \citenamefont {Keller}, \citenamefont
  {Buck}, \citenamefont {Jenny},\ and\ \citenamefont
  {Weber}}]{reichold_vascular_2009}%
  \BibitemOpen
  \bibfield  {author} {\bibinfo {author} {\bibfnamefont {J.}~\bibnamefont
  {Reichold}}, \bibinfo {author} {\bibfnamefont {M.}~\bibnamefont
  {Stampanoni}}, \bibinfo {author} {\bibfnamefont {A.~L.}\ \bibnamefont
  {Keller}}, \bibinfo {author} {\bibfnamefont {A.}~\bibnamefont {Buck}},
  \bibinfo {author} {\bibfnamefont {P.}~\bibnamefont {Jenny}},\ and\ \bibinfo
  {author} {\bibfnamefont {B.}~\bibnamefont {Weber}},\ }\bibfield  {title}
  {\bibinfo {title} {Vascular graph model to simulate the cerebral blood flow
  in realistic vascular networks},\ }\href
  {https://doi.org/10.1038/jcbfm.2009.58} {\bibfield  {journal} {\bibinfo
  {journal} {Journal of Cerebral Blood Flow \& Metabolism}\ }\textbf {\bibinfo
  {volume} {29}},\ \bibinfo {pages} {1429} (\bibinfo {year}
  {2009})}\BibitemShut {NoStop}%
\end{thebibliography}%

\newpage
\clearpage
\begin{center}\large{\textbf{Supplemental Material}}\end{center}
\normalsize
\section{Flow networks}
\label{sec:LinFlow}
In this section, we briefly review the theory and applications of linear flow networks.

\subsection{Mathematical description}
\label{sec:MathematicsLinFlow}

In this work, the main model of interest is a linear flow network model which we introduce more formally in this section. Consider a connected graph $G=(E,V)$ consisting of $N=|V|$ nodes and $L=|E|$ edges. Assign to each node in the network a potential $\vartheta_n\in\mathbb{R},~n\in V(G)$ and to each edge a capacity $K_{ij}\in\mathbb{R}^{+},~\ell=(i,j) \in E(G)$. Now we assign a \textit{flow} $F_{i\rightarrow j}\in\mathbb{R}$ to each link $\ell=(i,j)\in E(G)$ in the network that is assumed to be linear in the potential drop
\begin{equation}
    F_{i\rightarrow j}=K_{ij}\cdot(\vartheta_i-\vartheta_j) =-F_{j\rightarrow i}.
    \label{eq:linear_flow}
\end{equation}
 Suppose that there are sources and sinks attached to the nodes of the networks $P_i\in\mathbb{R},~i \in V(G)$. In this case, the in- and outflows at each node have to balance with the sources and sinks
\begin{equation}
    P_i=\sum_{k=1}^N F_{i\rightarrow k}.
    \label{eq:dc_approx}
\end{equation}
This equation is known as continuity equation or \textit{Kirchhoff's current law}. If the sources and sinks $P_i$ are given, Eqs.~\eqref{eq:dc_approx} and~\eqref{eq:linear_flow} completely determine the potentials in the network (up to a constant shift to all potentials). In a power grid, the sources and sinks are the power injections or withdrawals as a result of power production or consumption. When looking at the stable, operational fixed point of a power grid they are balanced such that $\sum_iP_i=0$ -- we therefore assume this to hold in the following sections. The theory of linear flow networks applies resistor networks, as well as AC power grids in the DC approximation, hydraulic networks and networks of limit cycle oscillators, which will be discussed in detail in section~\ref{sec:application}.

Now we introduce a compact, vectorial notation which facilitates the analysis of perturbations or damages to the network. Note that the flow is a signed quantity that depends on the orientation of the edges that we arbitrarily fix for this purpose and say that the flow is directed from node $i$ to node $j$ in this case. We can write the flows in vectorial notation $\vec{F}=(F_{1},...,F_L)^\top\in\mathbb{R}^{L}$ as follows;
\begin{equation}
    \vec{F}= \bm{K}\bm{I}^\top \vec{\vartheta}.
    \label{eq:flows_vectorial}
\end{equation}
Here, $\bm{K}=\operatorname{diag}(K_1,...,K_L)\in\mathbb{R}^{L\times L}$ is the graph's weight matrix that collects the edge weights and $\bm{I}^\top$ is the transpose of the the graph's edge-node incidence matrix $\bm{I}\in\mathbb{Z}^{N\times L}$ that determines the orientation of the graph's edges by the following relationship
\begin{equation}
    I_{j\ell} = \left\{
   \begin{array}{r l }
      +1 & \; \mbox{if edge $\ell  \,  \hat{=} \, (j,k)$ starts at node $j$},  \\
      -1 & \; \mbox{if edge $\ell  \,  \hat{=} \, (j,k)$ ends at node $j$ },  \\
      0     & \; \mbox{otherwise}.
  \end{array} \right. 
  \label{eq:ne_incidence}
\end{equation}
Furthermore, $\vec{\vartheta}=(\vartheta_1,...,\vartheta_N)^\top\in\mathbb{R}^N$ is a vector of potentials or voltage phase angles. We can also define a vector of power injections $\vec{P}=(P_1,...,P_N)^\top\in\mathbb{R}^N$ such that the continuity equation reads as
\begin{equation}
    \vec{P}=\bm{I}\vec{F}.
    \label{eq:Kirchhoff}
\end{equation}
In this expression, the correspondence between the power balance and \textit{Kirchhoff's current law} becomes manifest: it states that the in- and outflows at each node have to balance the injections and withdrawals of power. Combining Equations~\eqref{eq:flows_vectorial} and~\eqref{eq:Kirchhoff}, we may find a relationship between angles $\vec{\vartheta}$ and power injections $\vec{P}$, thus defining the graph's weighted \textit{Laplacian matrix} $\bm{L}=\bm{I}\bm{K}\bm{I}^\top\in\mathbb{R}^{N\times N}$, by
\begin{equation}
    \vec{P}=\bm{I}\bm{K}\bm{I}^\top\vec{\vartheta}=\bm{L}\vec{\vartheta}.
    \label{eq:vector_dc}
\end{equation}
The weighted Laplacian matrix used here has the following entries 
\begin{equation}
    L_{ij}   =\left\{\begin{array}{l l }
      -K_{\ell} & \; \mbox{if $i$ is connected to $j$ via $\ell = (i,j)$},  \\
      \sum_{\ell=(i,k)\in E(G)}K_{\ell} & \; \mbox{if $i=j$},  \\
      0     & \; \mbox{otherwise}.
  \end{array} \right. 
\end{equation}
The Laplacian matrix plays an important role in graph theory~\cite{newman2010}. If the underlying graph $G$ is connected, it has one zero eigenvalue $\lambda_1=0$ with corresponding eigenvector $\vec{v}_1=\vec{1}/\sqrt{N}$. Therefore, the matrix is not invertible. In many cases, it would nevertheless be desirable to invert the matrix, e.g. in order to find the phase variables given the power injections in Equation~\eqref{eq:vector_dc}. This problem is typically overcome by making use of the matrix's \textit{Moore-Penrose pseudoinverse} $\bm{L}^\dagger$. It may be used to invert Equation~\eqref{eq:vector_dc} in the same way as for the ordinary matrix inverse in the case of balanced power injections~\cite{van_mieghem_2017}. The Moore-Penrose pseudoinverse of the graph Laplacian $\bm{L}$ allows for the following representation: using $\bm L$'s eigenvalues sorted by magnitude $\lambda_1=0,\lambda_2\leq...\leq\lambda_N$ with corresponding eigenvectors $\vec{v}_1= \vec{1}/\sqrt{N},\vec{v}_2,...,\vec{v}_N$, we can express its pseudoinverse $\bm{L}^\dagger$ as \cite{Dorfler2018}
\begin{equation}
    \bm{L}^\dagger=\left(\vec{v}_1,\vec{v}_2,...,\vec{v}_N\right) 
    \begin{pmatrix} 
    0&0&...&0\\
    0&\lambda_2^{-1}&...&0\\
    ...&...&...&...\\
    ...&...&...&\lambda_N^{-1}
    \end{pmatrix}
    \left(\vec{v}_1,\vec{v}_2,...,\vec{v}_N\right)^\top\nonumber.
\end{equation}
The second eigenvalue $\lambda_2$ is usually referred to as \textit{Fiedler eigenvalue} or \textit{algebraic connectivity} and is an indicator of the graph's overall connectivity. If we assume the overall graph to be connected, this eigenvalue is strictly greater than zero $\lambda_2>0$. Importantly, a large difference between second and third eigenvalue $\lambda_3-\lambda_2$ implies a strong modularity in the graph and thus indicates the existence of a community structure~\cite{radicchi_defining_2004,girvan_community_2002,newman_communities_2012}.

Before we proceed, let us briefly fix the notation for the following sections: we will refer to an edge $\ell = (\ell_1,\ell_2)\in E(G)$ and its index $\ell$ in the ordered set of all edges interchangeably or refer to it by its terminal nodes $\ell_1$ and $\ell_2$. If we assume the edge space to be spanned by vectors in the two element field $GF(2)$, we may express the edge by a unit vector $\vec{l}_\ell=(0,...,\underbrace{1}_{l},..0)^\top\in GF(2)^{L}$ which we refer to as the edge's indicator vector. The edge-node incidence matrix $\bm{I}$ then maps this unit vector to the corresponding unit vectors in the field of vertices $GF(2)^N$. We thus get the following result for the edge expressed in terms of its starting vertex $\ell_1$ and terminal vertex $\ell_2$: 
\begin{equation}
    \vec{\nu}_\ell=\bm{I}\cdot \vec{l}_\ell = \vec{e}_{\ell_1}-\vec e_{\ell_2}=
    \begin{pmatrix}
    0\\
    ...\\
    1\\
    ...\\
    -1\\
    ...
    \end{pmatrix}
    \begin{matrix} \\ \\  \}\ell_1\\ \\ \}\ell_2\\  \\\end{matrix},\quad
    \nonumber
\end{equation}
where $\vec{e}_{\ell_1}$ and $\vec{e}_{\ell_2}$ are basis vectors in $GF(2)^N$
 \begin{equation}
 \begin{aligned}
    \vec{e}_{\ell_1} =    
    \begin{pmatrix}
    0\\
    ...\\
    1\\
    ...\\
    ...\\
    0
    \end{pmatrix}
    \begin{matrix} \\  \}\ell_1\\ \\ \\\end{matrix}
    ,\quad
        \vec{e}_{\ell_2} = 
        \begin{pmatrix}
    0\\
    ...\\
    ...\\
    1\\
    ...\\
    0
    \end{pmatrix}
    \begin{matrix} \\ \\ \\ \}\ell_2\\ \\  \\\end{matrix}
.\nonumber
    \end{aligned}
 \end{equation}
 This formulation allows us to easily switch between the edges expressed in edge space and the nodes corresponding to its terminal ends.

\subsection{Applicability of linear flow models}
\label{sec:application}
The theoretical framework in the last section has many different applications. We will demonstrate its applicability to the following systems in this section:
\begin{enumerate}
    \item Power grids,~\cite{Wood14,strake2018}
    \item Resistor networks~\cite{bollobas1998},
    \item Hydraulic networks~\cite{Sack2008,Kati10},
    \item Limit cycle oscillators~\cite{Manik2017}.
\end{enumerate}

\subsubsection{Application to power grids}
\label{sec:powergrids}
The power flow equations describing the steady state of a power system at an arbitrary node $i$ are given by~\cite{Wood14}
\begin{equation}
    \begin{aligned}
    P_i&=\sum_{k=1}^N |V_i||V_k|(G_{ik}\cos(\vartheta_i-\vartheta_k)+B_{ik}\sin(\vartheta_i-\vartheta_k))\\
    Q_i&=\sum_{k=1}^N|V_i||V_k|(G_{ik}\sin(\vartheta_i-\vartheta_k)-B_{ik}\cos(\vartheta_i-\vartheta_k))
    \label{eq:powerflow}.
\end{aligned}
\end{equation}
Here, $P_i$ and $Q_i$ are the real and reactive power generated or consumed at node or bus $i$, $\vartheta_i$ is the voltage angle at the same bus and $|V_i|$ is the voltage magnitude. The matrices $\bm{G}\in\mathbb{R}^{N\times N}$ and $\bm{B}\in\mathbb{R}^{N\times N}$ with elements $G_{ij}$ and $B_{ij}$, respectively, are the real part and the complex part of the complex nodal admittance matrix $\bm{Y}=\bm{G}+i\bm{B}\in\mathbb{C}^{N\times N}$. Note that the matrices $\bm{B}$ and $\bm{G}$ are not actually matrices of susceptances and conductances, respectively. Instead, their entries read as follows 
\begin{equation}
    \begin{aligned}
    B_{ij}   =\left\{\begin{array}{l l }
      -b_{ij} & \; \mbox{if $(i,j)\in E(G),~i\neq j$},  \\
      b_{i}^{\text{shunt}}+\sum_{(i,k)\in E(G)}b_{ik} & \; \mbox{if $i=j$},  \\
      0     & \; \mbox{otherwise},
  \end{array} \right.
\end{aligned}\nonumber
\end{equation}
where $b_i^{\text{shunt}}$ denotes the shunt susceptance of node $i$ and $b_{ij}$ is the susceptance of the circuit connecting node $i$ to node $j$. $\bm{G}$ has an analogous structure with elements
\begin{equation}
    \begin{aligned}
    G_{ij}   =\left\{\begin{array}{l l }
      -g_{ij} & \; \mbox{if $(i,j)\in E(G),~i\neq j$},  \\
      g_{i}^{\text{shunt}}+\sum_{(i,k)\in E(G)}g_{ik} & \; \mbox{if $i=j$},  \\
      0     & \; \mbox{otherwise},
  \end{array} \right.
\end{aligned}\nonumber
\end{equation}
where $g_{ij}$ are the conductances of the circuit between nodes $i$ and node $j$. The matrices $\bm{B}$ and $\bm{G}$ thus have the structure of a Laplacian matrix except for the diagonal entries which contain additional terms given by the shunt susceptances and conductances. The off-diagonal elements of the nodal admittance matrix thus read as 
\begin{equation}
    \begin{aligned}
    Y_{jk}=-y_{jk},~\forall j\neq k;\quad y_{jk}=g_{jk}+ib_{jk} = \frac{1}{r_{jk}+ix_{jk}},
\end{aligned}\nonumber
\end{equation}
with the circuit's reactance $x_{jk}$ and resistance $r_{jk}$. Note that line susceptances $b_\ell=\frac{-x}{r^2+x^2}$ are thus negative. The Equations~\eqref{eq:powerflow} reduce to the lossless power flow equations in the case where the real part of the nodal admittance matrix is negligible $\bm{G}\approx\bm{0}$, i.e. lines are purely inductive. 

We will focus on the so called \textit{DC approximation} of this full AC power flow equations. This approximation is based on three assumptions~\cite{Wood14}:
\begin{enumerate}
    \item Voltages vary little, i.e., $|V_i|\approx \text{const},~\forall i$ with respect to their base values,
    \item Angular differences are small, i.e., $\sin(\vartheta_i-\vartheta_j)\approx \vartheta_i-\vartheta_j,~\forall (i,j) \in E(G)$,
    \item Transmission lines are purely inductive, i.e., $B_{ij}\gg G_{ij},~\forall (i,j)\in E(G)$.
\end{enumerate}
Typically, these assumptions are fulfilled for high voltage transmission grids if the line loading is not too large~\cite{Purc05}. Using these approximation, Equation~\eqref{eq:powerflow} reduces to 
\begin{equation}
    \begin{aligned}
    P_i=\sum_{k=1}^N\underbrace{|V_i||V_k|B_{ik}}_{K_{ik}}(\vartheta_i-\vartheta_k),
\end{aligned}\nonumber
\end{equation}
thus revealing the analogy to Equation~\eqref{eq:dc_approx}. 

\subsubsection{Application to resistor networks}
Resistor networks are another example which may be described using linear flow networks. They have been studied for a long time leading to many fundamental results of graph theory~\cite{bollobas1998}. We will briefly introduce the theory of resistor networks and use the symbol $\hat{=}$ to refer to the corresponding quantity in the mathematical framework of linear flow networks as introduced in section~\ref{sec:MathematicsLinFlow}. For resistor networks, the flow along the graph's edges is a current flow $\vec{i}\in\mathbb{R}^L\hat{=}\vec{F}$ between nodes of different voltage $\vec{V}\in\mathbb{R}^N\hat{=}\vec{\vartheta}$. The line weights are given by the inverse resistances, i.e. the conductances, of the lines $\bm{G}\in\mathbb{R}^{L\times L}\hat{=}\bm{K}$ such that Equation~\eqref{eq:flows_vectorial} reads in this case
\begin{equation}
    \begin{aligned}
    \vec{i}=\bm{G}\bm{I}^\top\vec{V},
\end{aligned}\nonumber
\end{equation}
where $\bm{I}$ is again the node-edge incidence matrix. Along the same lines, Equation~\eqref{eq:Kirchhoff} translates to
\begin{equation}
    \begin{aligned}
    \vec{i}_{\textrm{in}} = \bm{I}\vec{i}.
\end{aligned}\nonumber
\end{equation}
Here, $\vec{i}_{\textrm{in}}\in \mathbb{R}^N\hat{=}\vec{P}$ is a vector of currents injected at the graphs' nodes and the Equation is again a manifestation of Kirchhoff's current law. We may thus apply the same theoretical framework to resistor networks.

\subsubsection{Applications to hydraulic networks}
The same formalism can also be shown to apply water transport networks that we refer to as hydraulic networks or pipe networks. Consider a hydraulic network consisting of pipes that connect to each other at junctions. Then we form the underlying graph by assigning a vertex to each of the junctions and put an edge between two vertices if they are connected via a pipe. The nodal quantity of interest in this case is the pressure $\vec{p}\in\mathbb{R}^N\hat{=}\vec{\vartheta}$. If we assume the pipes to be much longer than their radius $r\ll L$ and the flow across all pipes in the network to be laminar with a Newtonian, incompressible fluid flowing through it, we can approximate the fluid flow $\vec{Q}\in\mathbb{R}^L\hat{=}\vec{F}$ across a pipe $\ell=(i,j)$ by the \textit{Hagen-Poiseuille equation}
\begin{equation}
    \begin{aligned}
    Q_\ell=K_\ell\cdot (p_i-p_j).
\end{aligned}\nonumber
\end{equation}
Here, we collected different parameters describing the pipe and the fluid in the line parameter 
\begin{equation}
    \begin{aligned}
    K_\ell=\frac{\pi r_\ell^4}{8\mu L_\ell}
\end{aligned}\nonumber
\end{equation}
with the pipe radius $r_\ell$, the pipe length $L_\ell$ and the fluid's dynamic viscosity $\mu$. Conservation of mass then requires that inflows and outflows balance as in Equation~\eqref{eq:dc_approx}. Important applications of this framework are blood vessels in humans and animals~\cite{reichold_vascular_2009}, the vascular system of plants~\cite{Kati10} or hydraulic networks~\cite{wurbs2001}. For vascular networks, the system does not consist of pipes but rather of smaller vascular bundles such that the scaling of line parameter $K$ with the radius $r^4$ does not necessarily exactly hold~\cite{coomes_scaling_2008}. 

\subsubsection{Applications to limit cycle oscillators}
\label{sec:limit_cycle_oscillators}
The linear flow model may be regarded as a linearisation of the \textit{Kuramoto model} which naturally appears in many cases, in particular when approximating weakly coupled oscillator systems near a stable limit cycle~\cite{Manik2017}.

Consider a connected, simple graph $G=(E,V)$. The Kuramoto model describes a set of weakly coupled oscillators with phase angles $\vec{\vartheta}\in\mathbb{R}^N$ attached to the graph's vertices that are coupled via the graph's edges through coupling constants $K_{ij},(i,j) \in E(G)$,~see e.g.\cite{rodrigues_kuramoto_2016}. The oscillators' tendency to synchronise through the coupling is counteracted by each oscillator's natural frequency $\omega_j$ that is written compactly as a vector $\vec{\omega}=(\omega_1,...,\omega_N)^\top\in\mathbb{R}^N$. Then the dynamics of the phase angle $\vartheta_i$ attached to node $i,i\in\{1,...,N\}$, reads
\begin{equation}
    \begin{aligned}
    \dot{\vartheta}_i=\omega_i-\sum_k K_{ik} \sin(\vartheta_i-\vartheta_k).
\end{aligned}\nonumber
\end{equation}
As before, we fix an orientation of the graph's edges and summarise the coupling coefficients for all edges $(i,j)\in E(G)$ in the diagonal coupling matrix $\bm{K}\in\mathbb{R}^{L\times L}$, such that the vectorised dynamics reads
\begin{equation}
    \begin{aligned}
    \dot{\vec{\vartheta}}=\vec{\omega} -\bm{I}\bm{K}\sin(\bm{I}^\top\vec{\vartheta}).
    \label{eq:Kuramoto}
\end{aligned}
\end{equation}
Here, $\bm{I}$ is again the graph's node-edge incidence matrix~\eqref{eq:ne_incidence} and the sine function is understood to be taken element-wise, i.e 
\begin{equation}
    \begin{aligned}
    \sin(\bm{I}^\top\vec{\vartheta})=(\sin([\bm{I}^\top\vec{\vartheta}]_{1}),...,\sin([\bm{I}^\top\vec{\vartheta}]_{L}))^\top.
\end{aligned}\nonumber
\end{equation}
Fixed points of the dynamics are defined by a vanishing time derivative $\dot{\vec{\vartheta}}=\vec{0}$. Therefore, the equation characterising the phase angles at the fixed point $\vec{\vartheta}^*$ reads 
\begin{equation}
    \begin{aligned}
    \vec{\omega}=\bm{I}\bm{K}\sin(\bm{I}^\top\vec{\vartheta}^*).
\end{aligned}\nonumber
\end{equation}
If the angular differences on all edges are small, we may reduce this to the linear equation $\sin(\bm{I}^\top\vec{\vartheta})\approx\bm{I}^\top\vec{\vartheta}$, again retrieving an expression analogous to the discrete Poisson equation~\eqref{eq:vector_dc}. 

\subsubsection{The second-order Kuramoto model}
An extension of the Kuramoto model presented in Equation~\eqref{eq:Kuramoto} is given by the second-order Kuramoto model that is also used frequently in power systems analysis to describe synchronising generators~\cite{Rohden2012,nishikawa2015,manik_supply_2014}, where it is also referred to as Kuramoto model with inertia. The model contains an additional second-order time derivative of phase angles representing the generators' inertia and reads as
\begin{equation}
    \begin{aligned}
        \ddot{\vec{\vartheta}} =-\bm{\alpha}\dot{\vec{\vartheta}} +\vec{\omega} -\bm{I}\bm{K}\sin(\bm{I}^\top\vec{\vartheta}).
        \label{eq:kuramoto_second_order}
\end{aligned}
\end{equation}
Here, $\bm{\alpha}=\operatorname{diag}(\alpha_1,...,\alpha_N)\in\mathbb{R}^{N\times N}$ is a diagonal matrix incorporating the generators' inertia and friction coefficients~\cite{Rohden2012} and the other quantities are defined the same way as for the first order Kuramoto model~\eqref{eq:Kuramoto}. The vector of frequencies in this model corresponds to the power injections $\vec{\omega}\in\mathbb{R}^N\hat{=}\vec{P}$. Fixed points of the second order model with phase angles $\vec{\vartheta}^*$ are characterized by both, first and second order time derivative vanishing $\ddot{\vec{\vartheta}}=\dot{\vec{\vartheta}}=\vec{0}$ resulting in the same equation as for the first order model
\begin{equation}
    \begin{aligned}
    \vec{\omega}=\bm{I}\bm{K}\sin(\bm{I}^\top\vec{\vartheta}^*).
\end{aligned}\nonumber
\end{equation}
Again, this model reduces to the linear flow model if phase differences at the fixed point are small $\sin(\bm{I}^\top\vec{\vartheta}^*)\approx\bm{I}^\top\vec{\vartheta}^*$. 

\subsection{Description of link failures}
\label{sec:link_failures}
In this section, we will briefly review the analysis of link failures within the linear flow theory setting. We will first demonstrate how the effects of a link failure may be approached on the nodal level~\cite{strake2018}. Assume that a link $k=(r,s)$ with preoutage flow $\hat{F}_k$ fails, which does not disconnect the graph. This induces a change in the potentials 
\begin{equation}
    \begin{aligned}
    \vec{\vartheta}^\prime=\vec{\vartheta}+\Delta\vec{\vartheta}
\end{aligned}\nonumber
\end{equation}
by virtue of the discrete Poisson equation~\eqref{eq:vector_dc}. Here, we introduced the vector of potential changes $\Delta\vec{\vartheta}\in\mathbb{R}^N$ and a vector of potentials after the failure $\vec{\vartheta}^\prime\in\mathbb{R}^N$. The corresponding equation for the new grid reads as 
\begin{equation}
    \begin{aligned}
    \vec{P} = (\bm{L}+\Delta\bm{L})(\vec{\vartheta}+\Delta\vec{\vartheta}).
\end{aligned}\nonumber
\end{equation}
Here, $\Delta\bm{L}$ is the change in the Laplacian matrix due to the removal of link $k$ and takes the form $\Delta\bm{L}=K_k\bm{I}\vec{l}_k(\bm{I}\vec{l}_k)^\top$. If we subtract the discrete Poisson equation for the old grid before the failure of link $k$ from this equation, we arrive at the expression
\begin{equation}
    \begin{aligned}
    \Delta\vec{\vartheta} = -(\bm{L}+\Delta\bm{L})^\dagger\Delta\bm{L}\vec{\vartheta}.
\end{aligned}\nonumber
\end{equation}
Finally, we can use the Woodbury Matrix identity to rewrite the expression into the following form~\cite{strake2018}
\begin{equation}
    \begin{aligned}
    \bm{L} \Delta\vec{\vartheta} = q_{k} \vec \nu_{k},
    \label{eq:Poisson}
\end{aligned}
\end{equation}
where 
\begin{equation}
    \begin{aligned}
    q_{k}=(1-K_k(\bm{I}\cdot\vec{l}_{k})^\top \bm{L}^\dagger\bm{I}\cdot\vec{l}_k)^{-1}\hat{F}_k
\end{aligned}\nonumber
\end{equation}
is a source term. Similar expressions appear naturally when analysing resistor networks and have been studied, for example, in Refs.~\cite{van_mieghem_2017,Norman97}. After calculating the potential changes based on this equation, the flow changes on a link $\ell=(\ell_1,\ell_2)$ are given by the following equation 
\begin{equation}
    \begin{aligned}
    \Delta F_{\ell_1\rightarrow\ell_2}=K_\ell\cdot(\Delta \vartheta_{\ell_1}-\Delta \vartheta_{\ell_2}).
\end{aligned}\nonumber
\end{equation}

\section{Network isolators inhibit failure spreading completely}
\label{sec:digraph_isolator}

In this section we formally establish the existence of network isolators. To this end we first fix some notation.
\subsection{Fundamentals and notation}
We consider a linear flow network consisting of two parts, i.e.~its vertex set $V$ is written as $V = V_1 + V_2$. We now label the nodes in $V$ as follows without loss of generality
\begin{equation}
    \begin{aligned}
   & 1, \ldots, m_1:  & &\text{nodes in $V_1$ that are connected to $V_2$} \\
   & m_1+1,\ldots, n_1 : & &\text{nodes in $V_1$ that are not connected to $V_2$} \\
   & n_1+1, \ldots, n_1+m_2: & &\text{nodes in $V_2$ that are connected to $V_1$} \\
   & n_1+m_2+1, \ldots, n_1+n_2: & &\text{nodes in $V_2$ that are not connected to $V_1$}.
\end{aligned}\nonumber
\end{equation}
Then the weighted adjacency matrix of the network can be written as
\begin{equation}
    \begin{aligned}
   \bm{A} &= \begin{pmatrix}
             \bm{A}_1 & \bm{A}_{12} \\
             \bm {A}_{12}^\top & \bm{A}_2 
             \end{pmatrix} \\
   \bm{A}_{12} &=  \begin{pmatrix}
             \bm{a} & \bm{0} \\
             \bm{0} & \bm{0}
             \end{pmatrix}           
\end{aligned}\nonumber
\end{equation}
with $\bm{A}_1 \in \mathbb{R}^{n_1 \times n_1}$, $\bm{A}_2 \in \mathbb{R}^{n_2 \times n_2}$, $\bm{A}_{12} \in \mathbb{R}^{n_1 \times n_2}$ and $\bm{a} \in \mathbb{R}^{m_1 \times m_2}$. Furthermore, we define the degree matrices $\bm{D}_1$, $\bm{D}_2$ and $\bm{d}$ associated with the adjacency matrices $\bm{A}_1$, $\bm{A}_2$ and $\bm{a}$, that is
\begin{equation}
    \begin{aligned}
    d_{kl} = \left\{ \begin{array}{l l l}
         \sum_{p} a_{k p} & \text{for} & k=l \\
         0 & & k \neq l
    \end{array} \right. ,
\end{aligned}\nonumber
\end{equation}
and the Laplacian matrices $\bm{L}_1 = \bm{D}_1-\bm{A}_1$ of subnetwork 1, $\bm{L}_2 = \bm{D}_2-\bm{A}_2$ of subnetwork 2 and $\bm{L}$ of the whole system. 

\subsection{Main theorem on network isolators}
We can then formulate the main theorem on network isolators.

\begin{theorem}[Network isolators completely suppress failure spreading between modules]
\label{theo:weighted}
Consider a linear flow network consisting of two parts with vertex sets $V_1$ and $V_2$ and assume that a single link in the induced subgraph $G(V_1)$ fails, i.e.~a link $(r,s)$ with $r,s \in V_1$. If the adjacency matrix of the mutual connections has unit rank ${\rm rank}(\bm{A}_{12}) = 1$, then the flows on all links in the induced subgraph $G(V_2)$ are not affected by the failure, that is
\begin{equation}
    \begin{aligned}
   \Delta F_{\ell_1,\ell_2} \equiv 0 \quad \forall \ell_1,\ell_2 \in V_2.  
\end{aligned}\nonumber
\end{equation}
The subgraph corresponding to the mutual interactions is referred to as \textbf{network isolator}.
\end{theorem}

\begin{proof}
Assume that the adjacency matrix of the mutual connections has unit rank ${\rm rank}(\bm{A}_{12}) = {\rm rank}(\bm{a})  = 1$. We first proof that for any vector $\vec{y}\in \mathbb{R}^{n_1}$ the following statement holds
\begin{equation}
    \begin{aligned}
    \vec x = 
       \begin{pmatrix} 
           \bm{d}^{-1} \bm{a} & \bm{0} \\ \bm{0} & \bm{0} \end{pmatrix}
       \vec y
       = c \begin{pmatrix}
         1 \\ \vdots \\ 1 \\ 0 \\ \vdots \\ 0
         \end{pmatrix},
         \label{eq:rka-implies-c1}
\end{aligned}
\end{equation}
where $c\in\mathbb{R}$ is some real number. This result can be obtained by writing $\vec x\in\mathbb{R}^{n_2}$ in components. For all $j \in\{1,\ldots,m_2\}$ we have
\begin{equation}
    \begin{aligned}
    x_j = \frac{\sum_k a_{jk} y_k}{\sum_k a_{jk}}.
\end{aligned}\nonumber
\end{equation}
Since $\bm{a}$ has unit rank all its rows are linearly dependent such that we can write $a_{jk}/a_{1k} = a_{j1}/a_{11}$ for all $k \in \{1,\ldots,n_1\}$, such that $a_{jk}= a_{1k}a_{j1}/a_{11}$. Hence, 
\begin{equation}
    \begin{aligned}
    x_j &= \frac{a_{j1}/a_{11} \times \sum_k a_{1 k} y_k}{a_{j1}/a_{11} \times \sum_k a_{1k}} \\
    &= \frac{\sum_k a_{1 k} y_k}{\sum_k a_{1k}} 
    = x_1 =: c,
\end{aligned}\nonumber
\end{equation}
and all elements of the vector are equal. The remaining $n_2-m_2$ elements of the vector vanish, $x_j=0,~\forall j \in \{m_2+1,...,n_2\}$, because the corresponding adjacency matrix $\bm{A}_{12}$ has only zero entries at the respective positions.

We now compute the impact of a failure of link $k$ in $G(V_1)$ via the discrete Poisson equation~\eqref{eq:Poisson}
\begin{equation}
    \begin{aligned}
\bm{L}\Delta \vec{\vartheta}=q_k\vec\nu_k.
\end{aligned}\nonumber
\end{equation}
We decompose this equation as well as the vectors $\Delta \vec{\vartheta}$ and $\vec \nu$ into two parts corresponding to the two parts of the network
\begin{equation}
    \begin{aligned}
   \Delta \vec{\vartheta} = \begin{pmatrix} 
           \Delta \vec{\vartheta}_1 \\ \Delta \vec{\vartheta}_2 
        \end{pmatrix}, 
    \qquad \qquad
    \vec \nu = \begin{pmatrix} 
    \vec \nu_1 \\ \vec 0 
        \end{pmatrix} \, ,
\end{aligned}\nonumber
\end{equation}
where $\Delta \vec{\vartheta}_1,\vec{\nu}_1\in\mathbb{R}^{n_1}$ and $\Delta \vec{\vartheta}_2,\vec{\nu}_2\in\mathbb{R}^{n_2}$. Then the lower part of Equation~\eqref{eq:Poisson} corresponding to the vertices $n_1+1,\ldots, n_1+n_2$ reads
\begin{equation}
    \begin{aligned}
    \left[ \bm{L}_2  + 
    \begin{pmatrix} 
      \bm{d} & \bm{0} \\ \bm{0} & \bm{0}
    \end{pmatrix} \right] \Delta \vec{\vartheta}_2 = 
    \begin{pmatrix} 
      \bm{a} & \bm{0} \\ \bm{0} & \bm{0}
    \end{pmatrix} \Delta \vec{\vartheta}_1
    \label{eq:L2-d-a-psi1}
\end{aligned}
\end{equation}
using the notation established above. Using the prior result (\ref{eq:rka-implies-c1}) and multiplying by the matrix 
\begin{equation}
    \begin{aligned}
    \begin{pmatrix}
    \bm{d}^{-1} & \bm{0}\nonumber\\
    \bm{0} & \bm{1}
    \end{pmatrix}
    \end{aligned}
\end{equation}
this equation can be rewritten as
\begin{equation}
    \begin{aligned}
    \left[ 
    \begin{pmatrix} 
      \bm{d}^{-1} & \bm{0} \\ \bm{0} & \bm{1}
    \end{pmatrix} \bm{L}_2  + 
     \begin{pmatrix} 
           \bm{1} & \bm{0} \\ \bm{0} & \bm{0}
    \end{pmatrix} \right] \Delta \vec{\vartheta}_2 
    = \begin{pmatrix} 
    \bm{d}^{-1} \bm{a} & \bm{0} \\ \bm{0} & \bm{0} \end{pmatrix}
       \Delta \vec{\vartheta}_1
       = c \begin{pmatrix}
         1 \\ \vdots \\ 1 \\ 0 \\ \vdots \\ 0
         \end{pmatrix}
\end{aligned}\nonumber
\end{equation}
Now one can easily check via a direct calculation that 
\begin{equation}
    \begin{aligned}
    \Delta \vec{\vartheta}_2 = c \begin{pmatrix}
         1 \\ \vdots \\ 1 
         \end{pmatrix}
\end{aligned}\nonumber
\end{equation}
is a solution to this equation. Furthermore, this solution is unique as the linear system of equation has full rank. This is most easily seen for Equation (\ref{eq:L2-d-a-psi1}), as the matrix on the left hand side  is normal and positive definite.

We have thus shown that the nodal potentials in $V_2$ are shifted by the same constant $c$ when a link in $G(V_1)$ fails. Hence the flow changes are given by
\begin{equation}
    \begin{aligned}
   \Delta F_{\ell_1\rightarrow\ell_2} = K_{\ell}
   (\Delta\vartheta_{\ell_1} - \Delta \vartheta_{\ell_2} )
   = 0 \quad \forall \ell_1,\ell_2 \in V_2.  
\end{aligned}\nonumber
\end{equation}
\end{proof}

\begin{corollary}[Complete bipartite graphs are network isolators]
Consider a linear flow network consisting of two modules with vertex sets $V_1$ and $V_2$ and assume that a single link in the induced subgraph $G(V_1)$ fails, i.e.~a link $(r,s)$ with $r,s \in V_1$. If the subgraph $G^\prime$ of mutual connections between the two modules is a complete bipartite graph with uniform edge weights $K=K_\ell = K_m,~\forall \ell,m\in E(G^\prime)$, then the subgraph is a network isolator. If the whole graph is unweighted, $G^\prime$ always has uniform edge weights, thus a complete bipartite graph of mutual connections always is a network isolator for any unweighted network.
\end{corollary}
\begin{proof}
If the subgraph $G^\prime$ is complete and bipartite (ignoring all connections within both induced subgraphs $G(V_1)$ and $G(V_2)$), its adjacency matrix takes the form
\begin{equation}
    \begin{aligned}
    \bm{A}^\prime = K\cdot\begin{pmatrix}
             \bm{0} &\bm{1}_{m_1\times m_2} \\
             \bm{1}_{m_1\times m_2}^\top & \bm{0} 
             \end{pmatrix} .
\end{aligned}\nonumber
\end{equation}
We can immediately see that this matrix has unit rank, such that by theorem~\ref{theo:weighted}, $G^\prime$ is a network isolator.
\end{proof}

\subsection{Network isolators in non-linear systems}
We will now demonstrate how to extend the concepts of network isolators from linear systems to a certain class of non-linear networked systems with diffusive coupling. Let 
$$\vec{f}(\bm{L}\vec{x}) = (f_1([\bm{L}\vec{x}]_{1}),...,f_N([\bm{L}\vec{x}]_N))^\top:\vec{x}\in\mathbb{R}^N\rightarrow \vec{f}(\bm{L}\vec{x})\in\mathbb{R}^N$$
be a continuous function on the real numbers that depends on the product of Laplacian matrix $\bm{L}$ and vector $\vec{x}$.
Here, $[\bm{L}\vec{x}]_j$ denotes the $j$-th row of the standard matrix-vector product $\bm{L}\vec{x}$. We assume that the underlying network topology is again separated into two subgraphs $G(V_1)$ and $G(V_2)$, see the beginning of this section. We further assume that $$f_j(0)=0,~\forall j\in\{1,...N\},$$ 
i.e. each of the functions vanishes at the origin. 
Note that the functions $f_j([\bm{L}\vec{x}]_j)$ can be different and non-linear, as long as they vanish at the origin. Consider a dynamical system of the form 
\begin{equation}
    \dot{\vec{x}} = \vec{f}(\bm{L}\vec{x}) %
    \label{eq:nonlinear_dynamics}
\end{equation}
that admits a fixed point solution $\vec{x}^*$ with vanishing time derivative $\dot{\vec{x}}=\vec{0}$ that fulfills
\begin{equation}
    \vec{0} =\vec{f}(\bm{L}\vec{x}^*).
    \label{eq:nonlinear_fixedpoint}
\end{equation}
Now add a perturbation vector
\begin{equation}
    \begin{aligned}
    \Delta\vec{P}= \begin{pmatrix} 
           \Delta \vec{P}_1 \\ \vec{0}
        \end{pmatrix} 
    \end{aligned}
    \label{eq:nonlinear_perturbation}
\end{equation}
to the system that has non-zero entries only at the nodes of the first induced subgraph $G(V_1)$ and assume that the dynamical system~\ref{eq:nonlinear_dynamics} relaxes to a new fixed point $\vec{x}^\prime$ with
\begin{equation}
    \Delta\vec{P} =\vec{f}(\bm{L}\vec{x}^\prime).
    \label{eq:non_linear_new_fixedpoint}
\end{equation}
Then the following corollary holds
\begin{corollary}[Isolation in non-linear systems]
Consider a non-linear dynamical networked system of the form~(\ref{eq:nonlinear_dynamics}) that consists of two modules with vertex sets $V_1$ and $V_2$ which are connected by a network isolator as of Theorem~\ref{theo:weighted}. Assume that the system admits a fixed point solution as given in Eq.~\ref{eq:nonlinear_fixedpoint}. Assume that a small perturbation as in Eq.~\ref{eq:nonlinear_perturbation} is applied to the nodes in the first induced subgraph $G(V_1)$  and that the system relaxes to a new fixed point as in Eq.~\ref{eq:non_linear_new_fixedpoint}. Then the new fixed point has the following form
\begin{equation*}
    \vec{x}^\prime =
    \begin{pmatrix} 
           \vec{x}_1^\prime \\ c\vec{1}_2
    \end{pmatrix} ,
\end{equation*}
where $c\in\mathbb{R}$ is a constant. 
\end{corollary}
The second module is thus isolated against perturbations in the first module and vice versa in the sense that a perturbation in one module results in a constant shift in the other module.
\begin{proof}
The proof is analogous to the proof of theorem~\ref{theo:weighted}. Applying the function $\vec{f}$ to Equation~\ref{eq:L2-d-a-psi1} describing the fixed point in the non-perturbed subgraph $G(V_2)$, we see that the system is still solved by 
\begin{equation*}
    \vec{x}^\prime =
    \begin{pmatrix} 
          \Delta \vec{x}_1^\prime \\ c\vec{1}_2
    \end{pmatrix}.
\end{equation*}
\end{proof}

\subsection{Approximate isolation for weakly non-linear systems}
Even if not rigorously valid, we find that strong network isolation persists for an even larger class of non-linear systems that we will discuss in this section. Note that our analysis here closely follows a linear response theory analysis of Kuramoto oscillators that can be found in Ref.~\cite{Manik2017}. \\
Consider a networked non-linear dynamical system of the form
\begin{equation}
    \dot{\vec{x}}_i = \vec{f}(\vec{x}_i)_i+\sum_{k=1}^NA_{ik}g(x_i-x_k).\label{eq:non-linear_system}
\end{equation}
Here, $\vec{x}\in\mathbb{R}^N$ is a vector of nodal dynamical variables, $\vec{f}$ is a differentiable function of self-interactions of these variables and $\vec{g}(\vec{x})$ is a differentiable function that depends only on the differences of nodal variables at neighbouring nodes. The strength of interactions is encoded in the graph's adjacency matrix $\bm{A}$. Assume that the system relaxes to a fixed point with $\dot{\vec{x}}_i=0$ where $\vec{x}(t)=\vec{x}^*$. If we perturb the network locally at a node or an edge, we can compute the change in this fixed point using linear response theory~\cite{Manik2017}: to leading order, we obtain a linear system as above. \\
Assume that we perturb a single edge $(n,m)$ by modifying its edge weight by a small number $\Delta A_{ij}$ such that
\begin{align*}
    A_{ij}\rightarrow A_{ij}^\prime &= A_{ij}+\Delta A_{ij} \\
    \Delta A_{ij} &= \begin{cases} 0\qquad \text{if } (i,j)\neq (n,m)\\
    \Delta A\quad \text{if } (i,j)=(n,m)
    \end{cases}.
\end{align*}
Assume that this modification causes a change of the fixed point by
\begin{equation*}
    x_j^*\rightarrow x_j^\prime=x_j^*+\Delta x_j,\quad \forall j\in\{1,...,N\}.
\end{equation*}
We can expand the dynamics to leading order in terms of the new fixed point
\begin{equation*}
    \frac{\partial f(x^*_j)}{\partial x_j}\Delta x_j + \sum_{k=1}^NA_{jk}\frac{\partial g (x_j^*-x_k^*)}{\partial x_j}(\Delta x_j -\Delta x_k) + s_j = 0.
\end{equation*}
Here, $s_j$ is a source term that vanishes if node $j$ is not part of the edge $(n,m)$, $j\neq n,m$. The sum in this expression may be compactly written in terms of an effective Laplacian matrix $\tilde{\bm{L}}$
\begin{equation*}
    \sum_{k=1}^NA_{jk}\frac{\partial g (x_j^*-x_k^*)}{\partial x_j}(\Delta x_j -\Delta x_k) = [\tilde{\bm{L}}\Delta \vec{x}]_j,
\end{equation*}
where the Laplacian matrix has the off-diagonal entries 
$$\tilde{\bm{L}}_{jk} = -A_{jk}\frac{\partial g (x_j^*-x_k^*)}{\partial x_j}.$$
Thus, if the underlying graph contains a network isolator, we can apply Theorem~\ref{theo:weighted} to the system and see immediately that each component is (approximately) isolated against small perturbations in the other one. In particular, this description applies to Kuramoto oscillators (see section~\ref{sec:limit_cycle_oscillators}) perturbed at a few nodes or edges and powergrids described by AC load flow equations~\ref{eq:powerflow} subject to a link failure. We can thus get approximate isolation in both models as shown in Figure~4(e to h) for the AC load flow model.

\section{Linear controllability of complex networks}
\label{sec:linear_control}
We now turn to a different theoretical concept in complex networks research: the controllability of a network. In this section, we briefly analyse the influence of network isolators on the controllabilty of complex systems with a linear dynamics. In general, we find that introducing a network isolator to a complex network has no generic influence on its controllability.\\
Consider a linear dynamical system on a network with $N$ nodes with a state vector $\vec{x}\in\mathbb{R}^N$ whose dynamics is given by~\cite{yuan_exact_2013} 
\begin{align}
    \dot{\vec{x}} = \bm{A}\vec{x}+\bm{B}\vec{u}.
\end{align}
Here, $\bm{A}\in\mathbb{R}^{N\times N}$ denotes the graph's adjacency matrix, $\vec{u}\in\mathbb{R}^m$ is a (potentially time-varying) input vector that is supposed to achieve control of the network and $\bm{B}\in\mathbb{R}^{N\times m}$ is the control matrix. Then one definition of controllability is the following: Can we find a set of $m$ driver nodes identified by the controllability matrix $\bm{B}$ such that the system may be driven from any initial state $\vec{x}_0$ to any final state $\vec{x}_f$ in finite time? If yes, the system is said to be \textit{controllable} and a measure of its controllability is given by the minimum number of driving nodes $N_d\leq N$ necessary to achieve full controllability~\cite{yuan_exact_2013,gao_target_2014,liu_controllability_2011}. \\
We identify this set of driver nodes necessary for exact controllability for a small sample network using a method due to Yuan et al.~\cite{yuan_exact_2013} who demonstrated that the minimum number of driver nodes $N_d$ can be found by determining the multiplicity of the eigenvalues of the graph's adjacency matrix $\mathbf{A}$~\cite{yuan_exact_2013}. Assume that the underlying network is undirected such that its adjacency matrix is symmetric as for the networks studied in this manuscript. In this case, we can calculate the algebraic multiplicity $\delta(\lambda_i)$ for all eigenvalues $\lambda_i$ of this matrix to calculate the minimum number of driver nodes, $N_D$, necessary to control the network (cf. Eq.4, Ref.~\cite{yuan_exact_2013})
\begin{align}
    N_D = \operatorname{max}_i\left[\delta(\lambda_i)\right]\label{eq:controllability_yuan}.
\end{align}
This approach has the advantage that the driver nodes necessary to control the network, i.e. the controllability of a network, may immediately be identified, which is more complicated when using the classical Kalman rank condition~\cite{yuan_exact_2013}.\\
In Figure~\ref{fig:isolator_controllability}, we illustrate a potential application of this formalism to network isolators. The adjacency matrix of the graph reported in panel A has the eigenvalue $\lambda_M=-1$ with multiplicity $\delta(\lambda^M) = 2$, while all other eigenvalues have multiplicity one. An eigenvalue $\lambda_M = -1$ in the adjacency matrix can easily be constructed by connecting two nodes to the other nodes in a network in exactly the same way~\cite{van_mieghem_2017}. Thus, by the criterion~\ref{eq:controllability_yuan}, only two nodes are required to control the network. These nodes have been determined using the method described in Ref.~\cite{yuan_exact_2013} and are highlighted in orange. After introducing the isolator into the system (panel D), the maximum multiplicity of any eigenvalue of the graph's adjacency matrix is one, i.e.  $\delta(\lambda_i)=1,~\forall i$, which implies that the graph can be controlled by a single node (colored red). Therefore, in this case, the controllability of the network is increased after constructing the isolator. We emphasize that the network isolator prevents only flow changes, but not flows from passing as demonstrated in panels B-C and E-F. \\
For the remaining network isolators constructed in throughout this manuscript, we did not find any influence of the introduction of network isolators on the controllability of the underlying network and thus conclude that isolators do not generically influence network controllability.

\section{Computational methods}

\subsection{Creating graphs with strong or weak inter-module connectivity}
\label{sec:graph_model}

We introduce a model to create ensembles of graphs consisting of two subgraphs with weak or strong interconnectvity, see Figures~1 and 2. We start with two disconnected Erd\H{o}s-R\'{e}nyi random graphs $G_1(N_1,p_1)$ and $G_2(N_2,p_2)$, where $N$ denotes the number of nodes in the grid and $p$ the probability that two randomly chosen nodes are connected by an edge \cite{Erdos1960}. Then we randomly choose $n_1=[c\cdot N_1]$ nodes $v=\{v_1,. . .,v_{n_1}\}$ in $G_1$ and $n_2=[c\cdot N_2]$ nodes $w=\{w_1,. . .,w_{n_2}\}$ in $G_2$. Here, $c\in [0,1]\subset\mathbb{R}$ is a constant representing the share of nodes connecting to the other subgraph and $[\cdot]$ denotes the nearest integer. Out of all possible edges $e=\{(v_1,w_1),.. .,(v_{n_1},w_1),.. . ,(v_{n_1},w_{n_2})\}$ between the two sets of nodes $v$ and $w$, we randomly add a share of $\mu\in[0,1]$. The parameter $\mu$ controls the connectivity of the two subgraphs $G_1$ and $G_2$: They remain disconnected for $\mu=0$ and they are connected via a complete bipartite graph for for $\mu=1$. For $c=1$ and $\mu=p_1=p_2$ we recover a single Erd\H{o}s-R\'{e}nyi random graph with $N=N_1+N_2$ nodes. Note that this procedure is in principal not limited to ER random graphs. We apply it to study other types of graphs as shown in Supplemental Figure~\ref{fig:ratio_links_and_d}.

\subsection{Perturbing network isolators}

The robustness of network isolators to structural perturbations is analysed as follows. Let $G=(E,V)$ be a graph whose nodes are split into two subsets $V_1$ and $V_2$. Furthermore, let $\mathbf{A}_{12}$ be the part of the graph's weighted adjacency matrix that encodes the mutual connections between the two parts as described in theorem~\ref{theo:weighted}. Without loss of generality we can order the nodes of the network in such a way that the matrix has the structure
\begin{equation}
    \mathbf{A}_{12}=
    \begin{pmatrix}
    \vec{a}_{1} & \cdots & \vec{a}_{m} & \vec{0} & \cdots & \vec 0\\
    \vec{0} & \cdots & \vec{0} & \vec{0} & \cdots & \vec 0
    \end{pmatrix}.
\end{equation}
According to theorem~\ref{theo:weighted}, a perfect network isolator is found if $\operatorname{rank}(\mathbf{A}_{12})=1$, i.e.~if all vectors $\vec a_1, \ldots, \vec a_m$ are linearly dependent. 

To investigate the robustness of network isolators, we start from a unit rank matrix $\operatorname{rank}(\mathbf{A}_{12})=1$ and perturb it iteratively. In each step we choose one of the vectors $\vec a_i, i=1,\ldots,m$ at random and perturb it according to $\vec{a}_i^\prime = \vec{a}_i+ \vec{e} \lVert\vec{a}_i\rVert$. The elements of the perturbation vector $\vec{e}$ are chosen uniformly at random from the interval $[-\beta,\beta]$, where $\beta$ is a small parameter, here $\beta=0.05$.

The deviation of the perturbed matrix $\mathbf{A}_{12}$ from a unit rank matrix is quantified using its coherence statistics~\cite{tropp_computational_2010}, 
\begin{equation}
    \xi(\mathbf{A}_{12}) = 1-\min_{i,j}\frac{\langle \vec{a}_i,\vec{a}_j\rangle}{\lVert\vec{a}_i\rVert\lVert\vec{a}_j\rVert},
\end{equation}
where $\langle\cdot,\cdot\rangle$ denotes the standard scalar product on $\mathbb{R}^n$ and $\lVert \cdot \rVert$ denotes the $\ell^2$-norm. For a matrix $\mathbf{A}_{12}$ of unit rank we have $\xi(\mathbf{A}_{12})=0$ as all vectors are linearly dependent. For vectors deviating from linear dependence, the measure increases until it reaches its maximum value if two vectors are linearly independent with $\xi(\mathbf{A}_{12})=1$. 

To create Figure~2(e), we repeated this process $1000$ times starting from the perfect isolator shown in panel c. Edge weights were randomly chosen from a normal distribution $\mathcal{N}(10,1)$ with mean $\mu=10$ and variance $\sigma^2=1$ except for the isolator, where we choose four groups of four edges having the same weight such that initially $\operatorname{rank}(\mathbf{A}_{12})=1$. For each perturbed network, we evaluate $\xi(\mathbf{A}_{12})$ and the ratio of flow changes $R$ according to Eq.~(3) averaged over all possible trigger links $\ell$ and distances $d$. For a perfect isolator, this ratio vanishes due to a vanishing numerator.

\subsection{Power grid data and cascade model}
\label{sec:netdata}

Power grid data has been extracted from the open European energy system model PyPSA-Eur, which is fully available online \cite{horsch_2018}. The model includes the topology as well as the susceptance $b_\ell$ and the line rating $F_{i\rightarrow j}^{\rm max}$  for each high voltage transmission line in Europe. We consider the Scandinavian synchronous grid spanning Norway, Sweden, Finland and parts of Denmark. This grid is coupled to other synchronous grids (central European grid, British grid and Baltic grid) only via high voltage DC transmission lines. Power flow on these lines are actively controlled and can thus be considered constant, thus leading to constant real power injections at the coupling nodes. The Scandinavian grid has 269 nodes and 370 edges, counting multiple-circuit lines only once. 

Cascading failures are simulated for fixed power injections $P_i$ for each node corresponding to an economic dispatch for the entire PyPSA-Eur model that includes a security margin given by the constraint $|F_{i\rightarrow j}| \leq 0.8\cdot F_{i\rightarrow j}^{\rm max}$. The cascade is triggered by the failure of a single line $(r,s)$ which is effectively removed from the grid. The simulation then proceeds step-wise; In each step, we first calculate the nodal phase angles $\vartheta_i$ and real power flows $F_{i\rightarrow j}$ for all nodes and lines, respectively, by solving the continuity equation  $P_i = \sum_j F_{i\rightarrow j}$ with $F_{i\rightarrow j} = K_{ij} (\vartheta_i-\vartheta_j)$. Then we check for overloads: Any line $(i,j)$ with $|F_{i\rightarrow j}| > F_{i\rightarrow j}^{\rm max}$ undergoes an emergency shutdown and is removed from the grid. The simulations are stopped when no further overload occurs or when the grid is disconnected.

Note that this mechanism for cascading failures is different from the cascading failure mechanism typically analysed in node capacity load models (see e.g. Refs.~\cite{motter_ml_2002,crucitti_model_2004}). The redistribution of nodal loads or flows after failures in such models is typically based on neighborhood of nodes, on shortest path betweenness measures or on other `intelligent' redistribution schemes whereas the redistribution of flows after failures in linear flow networks or power grids studied using AC load flow analysis are given by the physical laws governing electrical networks. Furthermore, usually nodes - not edges - are assumed to fail, which is not the typical case in real power grids.

\subsection{Processing leaf data}
The leaf venation network is based on a microscopic recording of a leaf of the species \textit{Bursera hollickii} provided by the authors of Ref.~\cite{ronellenfitsch_topological_2015}. Edge weights $K_{ij}$ are assumed to scale with the radius $r_{ij}$ of the corresponding vein $(i,j)$ as $K_{ij}\varpropto r_{ij}^4$ according to the Hagen-Poisseuille law, see Ref.~\cite{coomes_scaling_2008} for a detailed discussion. We used the radius in pixel scanned at a resolution of 6400 dpi.

\makeatletter 
\renewcommand{\thefigure}{S\@arabic\c@figure}
\makeatother 
\setcounter{figure}{0}
\begin{figure*}[ht!]
    \begin{center}
        \includegraphics[width=\textwidth]{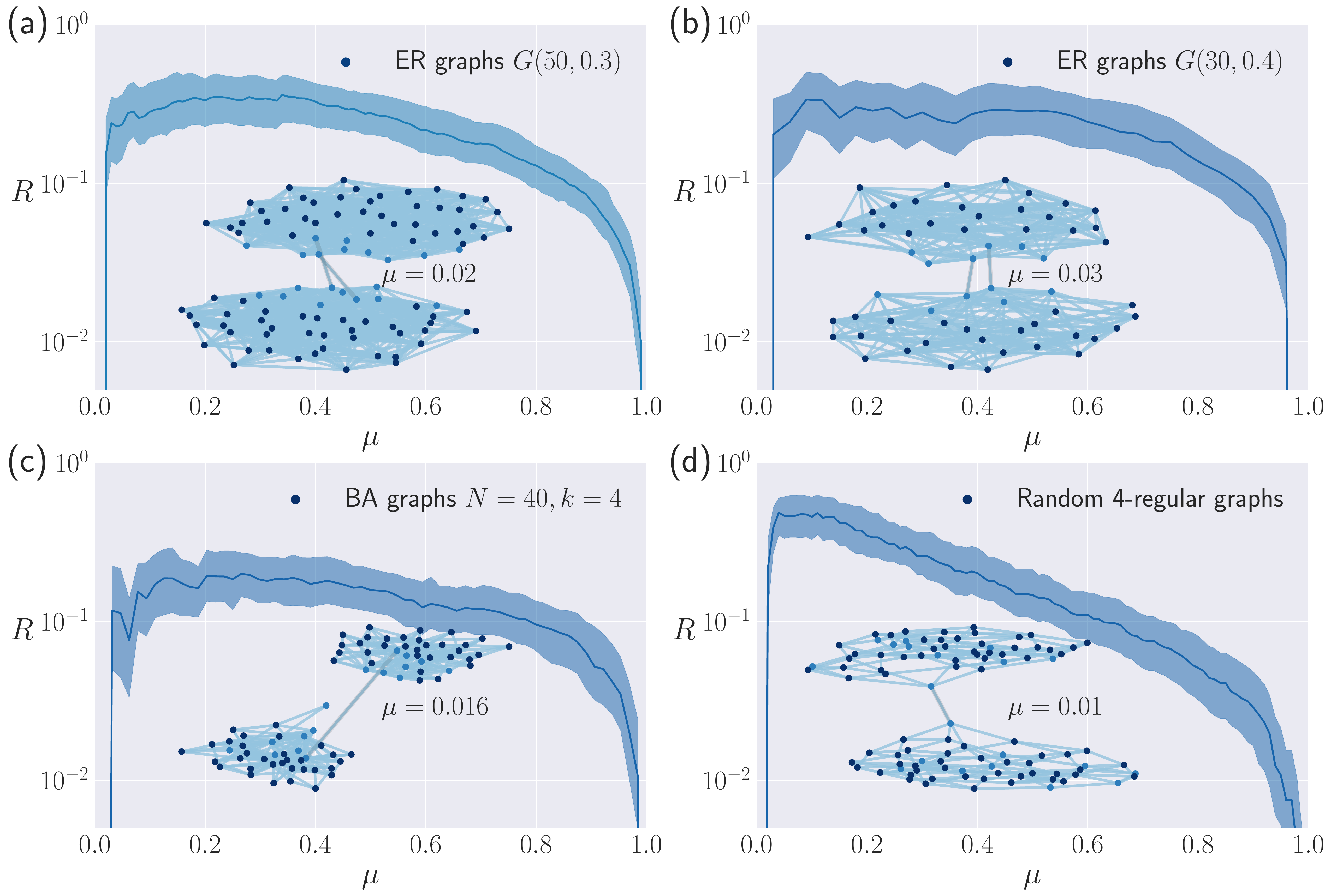}
    \end{center}
    \caption{
    \textbf{Averaged ratio of flow changes decays with high and low connectivity for different random graphs.} All panels show ratio of flow changes $R$ averaged over all links and distances against connectivity parameter $\mu$ (see methods) along with corresponding graph for low values of the connectivity parameter. (a) Two ER graphs with parameters $N_i = 50, p_i = 0.3$ connected with probability $\mu = 0.02$ at a randomly chosen share of $c = 0.2$ their nodes. (b) Same as in (A), but with parameters $N_i = 30, p_i = 0.4, \mu = 0.03, c = 0.2$. (c) A similar scaling is observed if two BA random graphs with parameters $N_i = 40, k_i=4$ are connected with probability $\mu = 0.016$ at a randomly chosen share of $c = 0.2$ their nodes. (d) The scaling is also preserved if two 4-regular, random graphs are connected with parameters $N=50,\mu=0.01,c=0.2$. Blue line represents median value over all distances and shaded region indicates $0.25$- and $0.75$-quantiles for all graphs. 
    \label{fig:ratio_links_and_d}
    }
\end{figure*}
\begin{figure*}[ht!]
    \begin{center}
        \includegraphics[width=\textwidth]{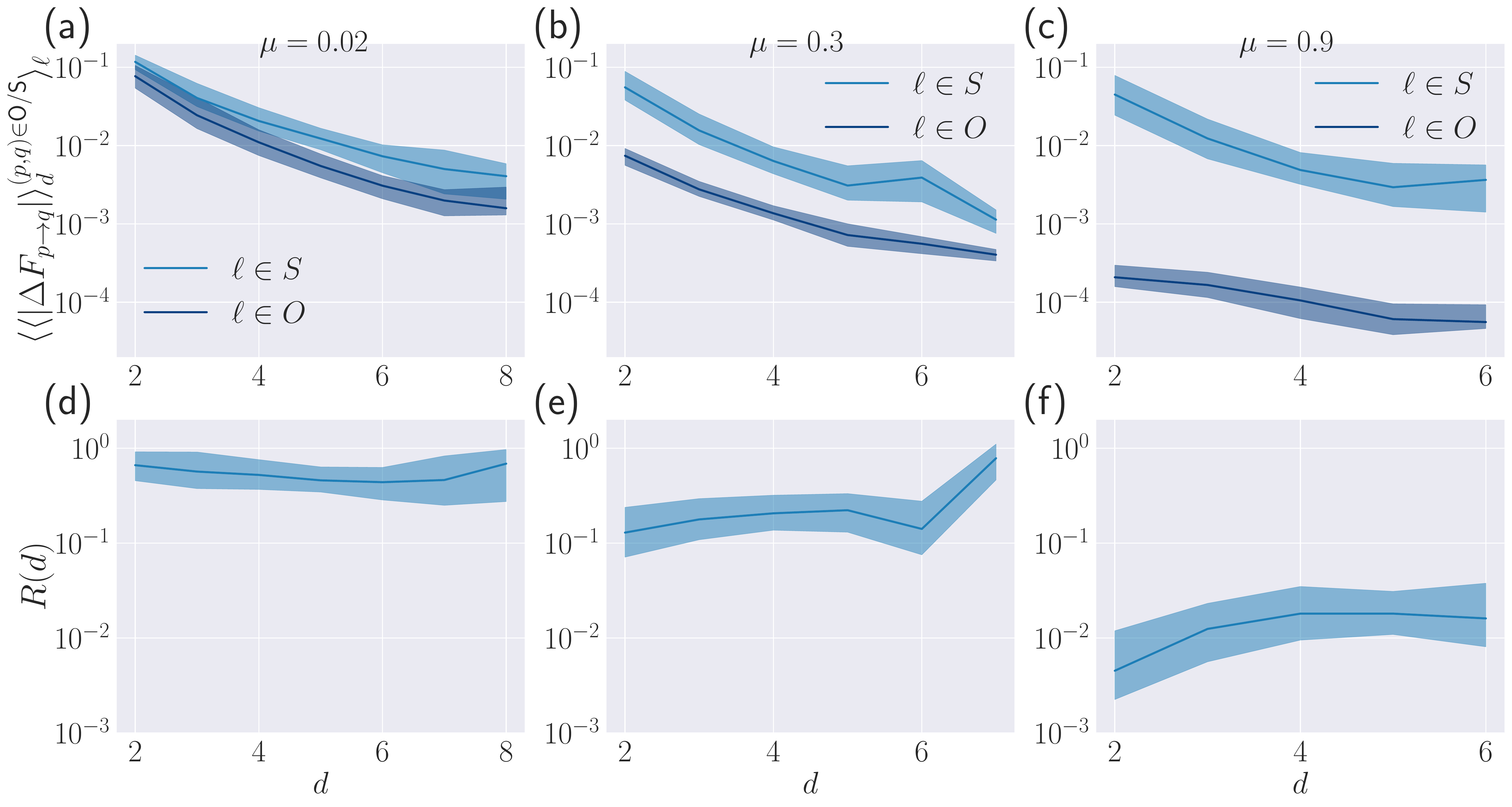}
    \end{center}
    \caption{
    \textbf{Ratio of flow changes depends weakly on distance.} We examine the scaling of link flow changes with distance for two ER random graphs $G(120,0.02)$ that are connected at $c=0.2$ nodes with changing probabilities $\mu=0.02$ (left), $\mu=0.3$ (centre) and $\mu=0.9$ (right). We only consider the largest component from each of the two random graphs and remove all dead ends as they result in vanishing flow changes. (a to c) Normalised absolute flow changes decay with distance when averaging over all possible trigger links. We always assume a unit flow on the failing link before the failure. We distinguish flow changes in the same (blue, top) and the other (purple, bottom) module of the graph. Flow changes are consistently higher in the same module for all distances. (d to f) Ratio of flow changes averaged over all possible trigger links $R(d)$ reveals a weak dependence of the ratio on distance. Blue line represents median value over all distances and shaded region indicates $0.25$- and $0.75$-quantiles for all graphs. 
    \label{fig:ratio_links}
    }
\end{figure*}
\begin{figure*}[ht!]
    \begin{center}
        \includegraphics[width=\textwidth]{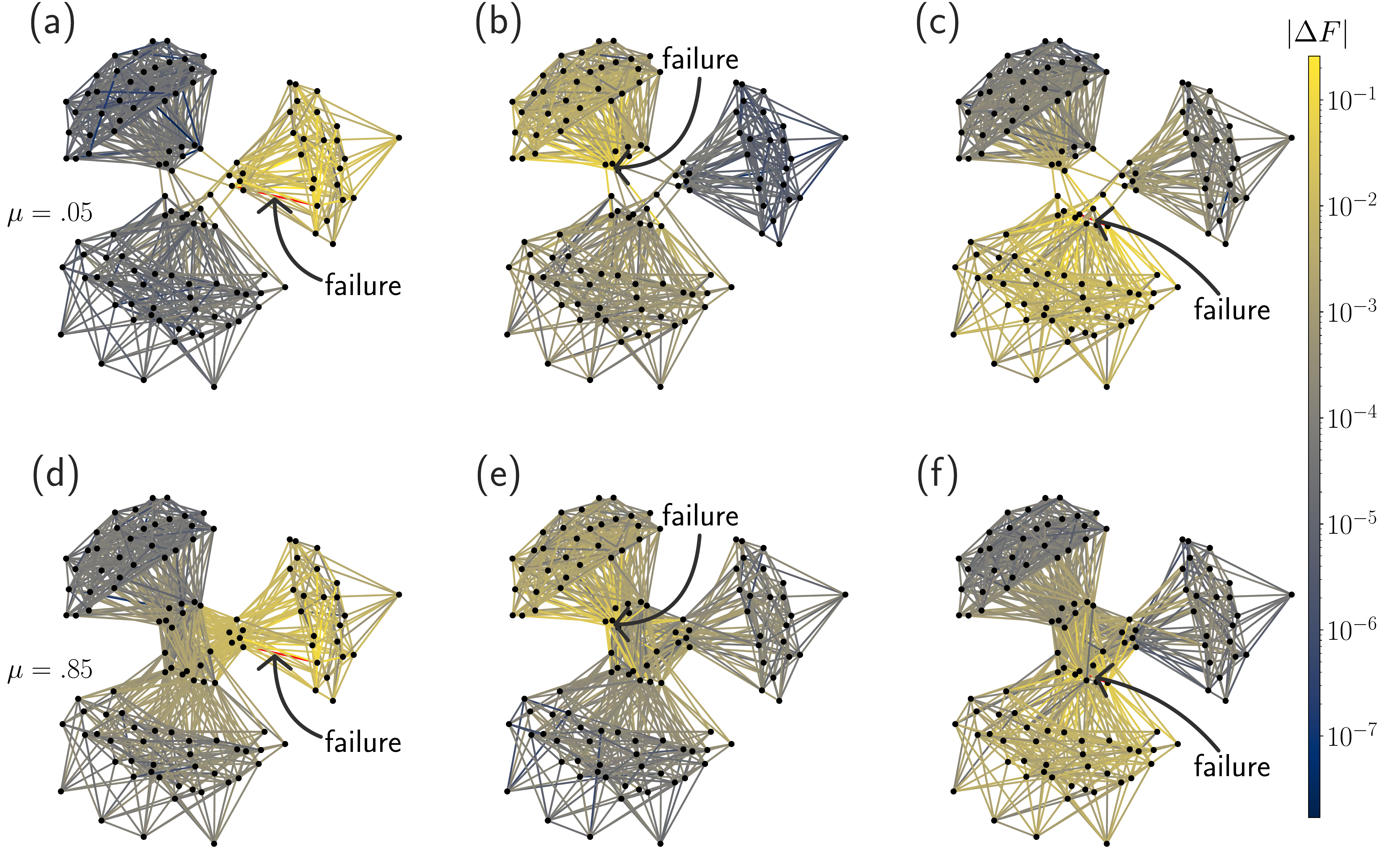}
    \end{center}
    \caption{
    \textbf{Increasing or decreasing connectivity between more than two modules reduces failure spreading equally well.} Here, we demonstrate a possible extension of the synthetic network model described in the Methods section to more than two modules. For each panel, we simulate a single link failure (red) that results in flow changes (colour coded). (a to c) Three ER random graphs $G(30,0.3)$ (right), $G(50,0.2)$ (bottom) and $G(40,0.4)$ (top left) that are mutually interconnected with probability $\mu=0.05$ at 20 percent, i.e. $c=0.2$, thus resulting in three mutually weakly connected modules. (d to f) Connecting the same modules as shown in (A to C) with probability $\mu = 0.85$, thus resulting in strong inter-module connectivity, reduces failure spreading equally well.
    \label{fig:connectivity_three_parts}
    }
\end{figure*}

\begin{figure*}[ht!]
    \begin{center}
        \includegraphics[width=\textwidth]{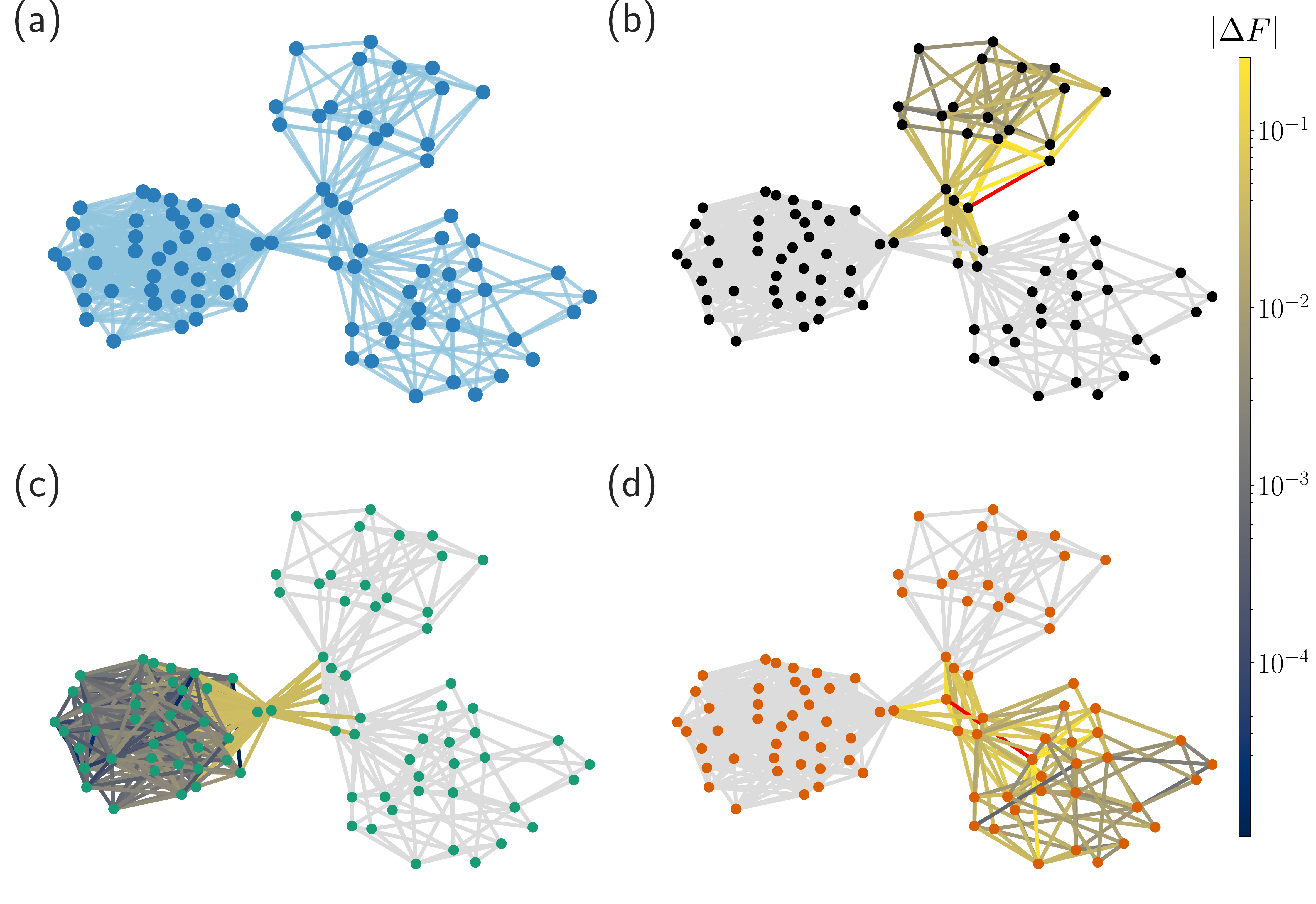}
    \end{center}
    \caption{
    \textbf{Networks isolators can be generalised to network consisting of more than two modules.}
    (a) Topology of a network consisting of three ER random graphs $G(40,0.4)$ (left), $G(20,0.3)$ (top) and $G(30,0.2)$ (bottom right) that are mutually connected through network isolators.
    (b to d) Link failures in each of the individual subgraphs (red lines) do not change flows (colour code) in any of the other subgraphs.
    \label{fig:isolator_three_parts}
    }
\end{figure*}

\begin{figure*}[ht!]
    \begin{center}
        \includegraphics[width=\textwidth]{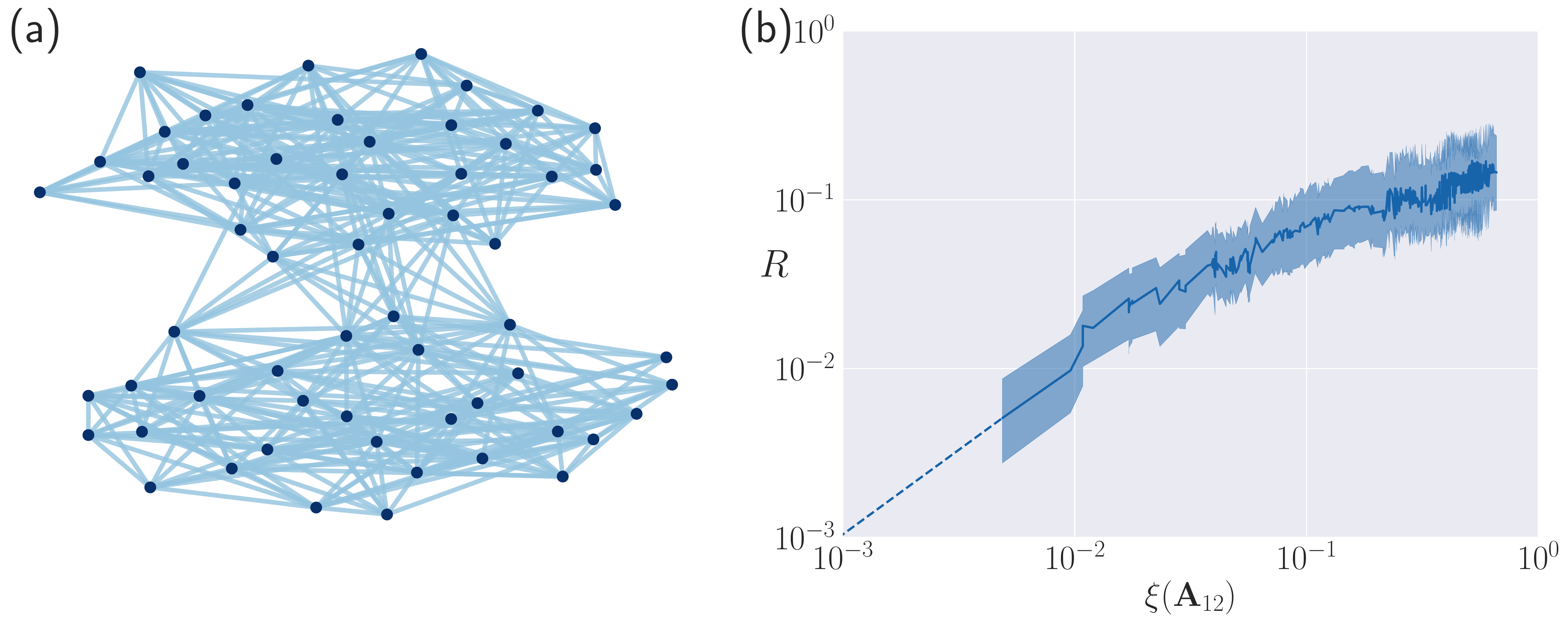}
    \end{center}
    \caption{
    \textbf{Robustness of network isolators shows the same scaling with perturbations for different graphs.}
    Robustness of network isolators measured by ratio of flow changes $R$ averaged over all links against measure of perturbations to network isolators $\xi(\mathbf{A}_{1,2})$. (a) Graph created from the graph ensemble and shown in Fig.~\ref{fig:ratio_links_and_d}C was modified in such a way that it contains a network isolator connecting five nodes from one part to five nodes of the other part through a bipartite connectivity structure. Edge weights are drawn randomly from a normal distribution $\mathcal{N}(10,1)$ except for the network isolator where the randomly chosen weights of five edges starting in the same node and connecting to all connecting nodes in the other part were chosen as basis weights for all other connections between the two parts. (b) The isolator robustness shows qualtiatively the same scaling as for the 6-regular graph shown in Fig.~1(c). Perturbations were applied in 1000 repetitions choosing a perturbation strength of $\alpha=0.05$. Dotted line takes into account the fact that the curve goes through the point $\xi=R=0$ for a perfect isolator.
    \label{fig:isolator_robustness_other_graph}
    }
\end{figure*}

\begin{figure*}[h!]
    \begin{center}
        \includegraphics[width=\textwidth]{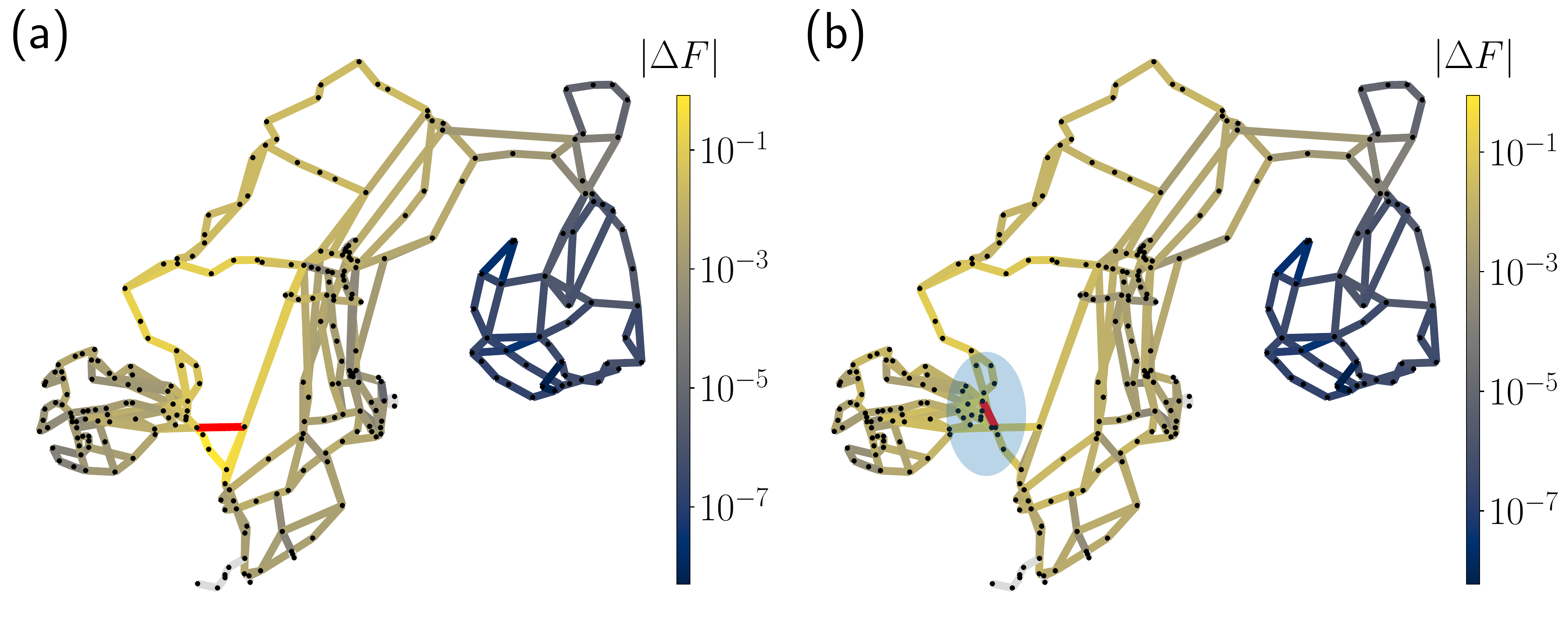}
    \end{center}
    \caption{
    \textbf{Network isolators do not increase grid vulnerability.}
    (a) Failure of a link with unit flow in the Scandinavian grid before the construction of the network isolator yields a strong response in terms of absolute flow changes $|\Delta F|$. (b) After adding two links to create a network isolator (blue shaded region, see Figure~3(c)), we simulate a failure of one of the links \emph{in} the isolator. We observe that both, the failure within the isolator (b) as well as a failure in the initial grid in close proximity to the location where the isolator is constructed (a) yield a similar effect. We thus conclude that introducing the network isolator will not make the network more vulnerable compared to the network without the isolator. However, a failure in the isolator may potentially affect the whole network.
    \label{fig:failure_in_isolator}
    }
\end{figure*}

\begin{figure*}[h!]
    \begin{center}
    \includegraphics[width=0.8\textwidth]{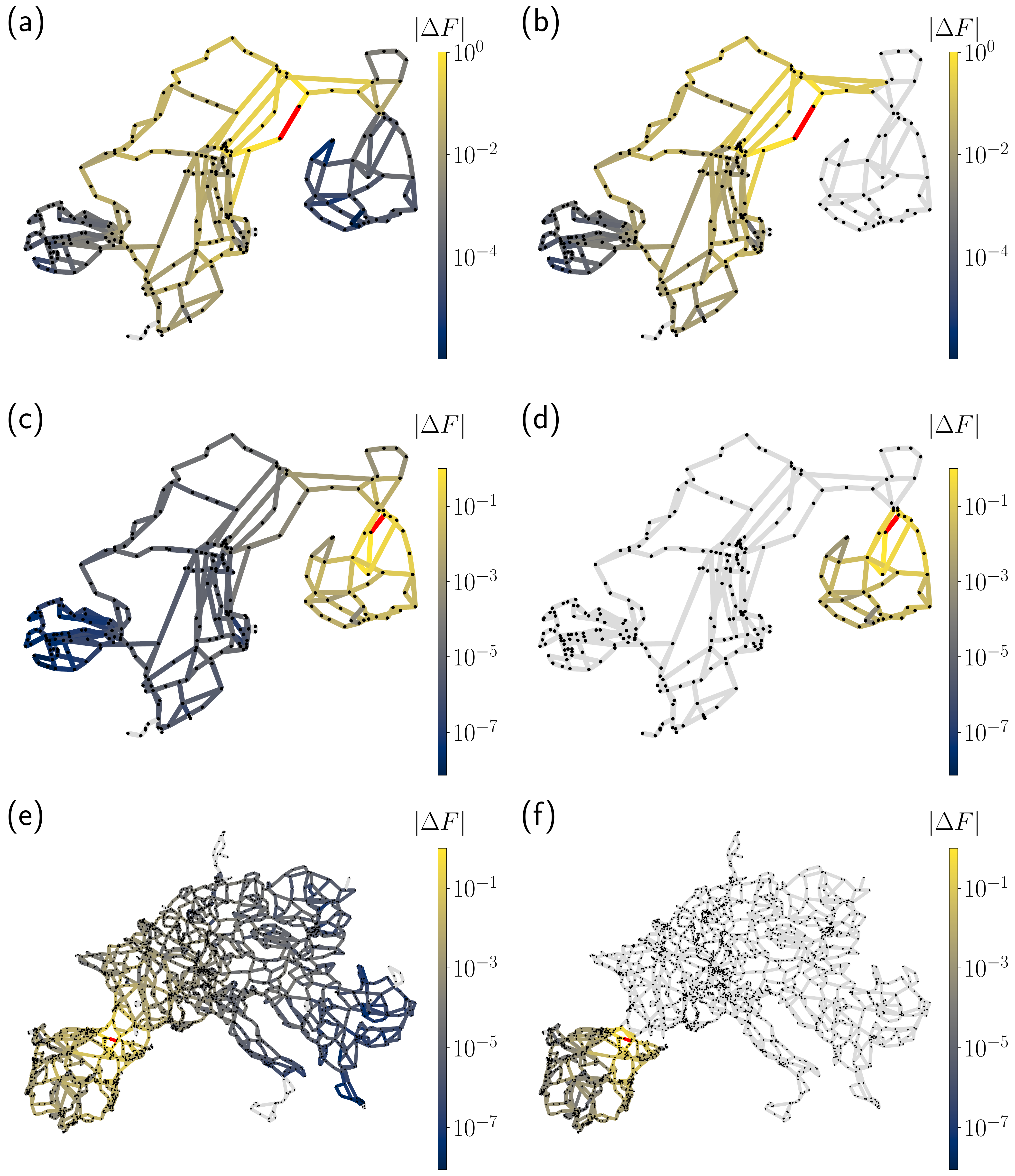}
    \end{center}
    \caption{
    \textbf{Network isolators may be realised in various real-world power grids.} All grid topologies and line susceptances were extracted from the open European energy system model PyPSA-Eur, which is fully available online\cite{horsch_2018}. (a,c,e) Initial failure of a link (red) with unit flow results in flow changes in the whole network for Scandinavia (a,c) as well as the central European grid (e). (b,d,f) After introducing network isolators to the grids, failure spreading to other parts of the network is completely stopped. The construction of isolators follows the ``recipes'' illustrated in Figure~4, namely with (b) following the construction in Fig.~4(b), (d) following the construction in Fig.~4(a), and (f) following the construction in Fig.~4(d).%
    \label{fig:isolator_examples}
    }
\end{figure*}

\begin{figure*}[h!]
    \begin{center}
    \includegraphics[width=\textwidth]{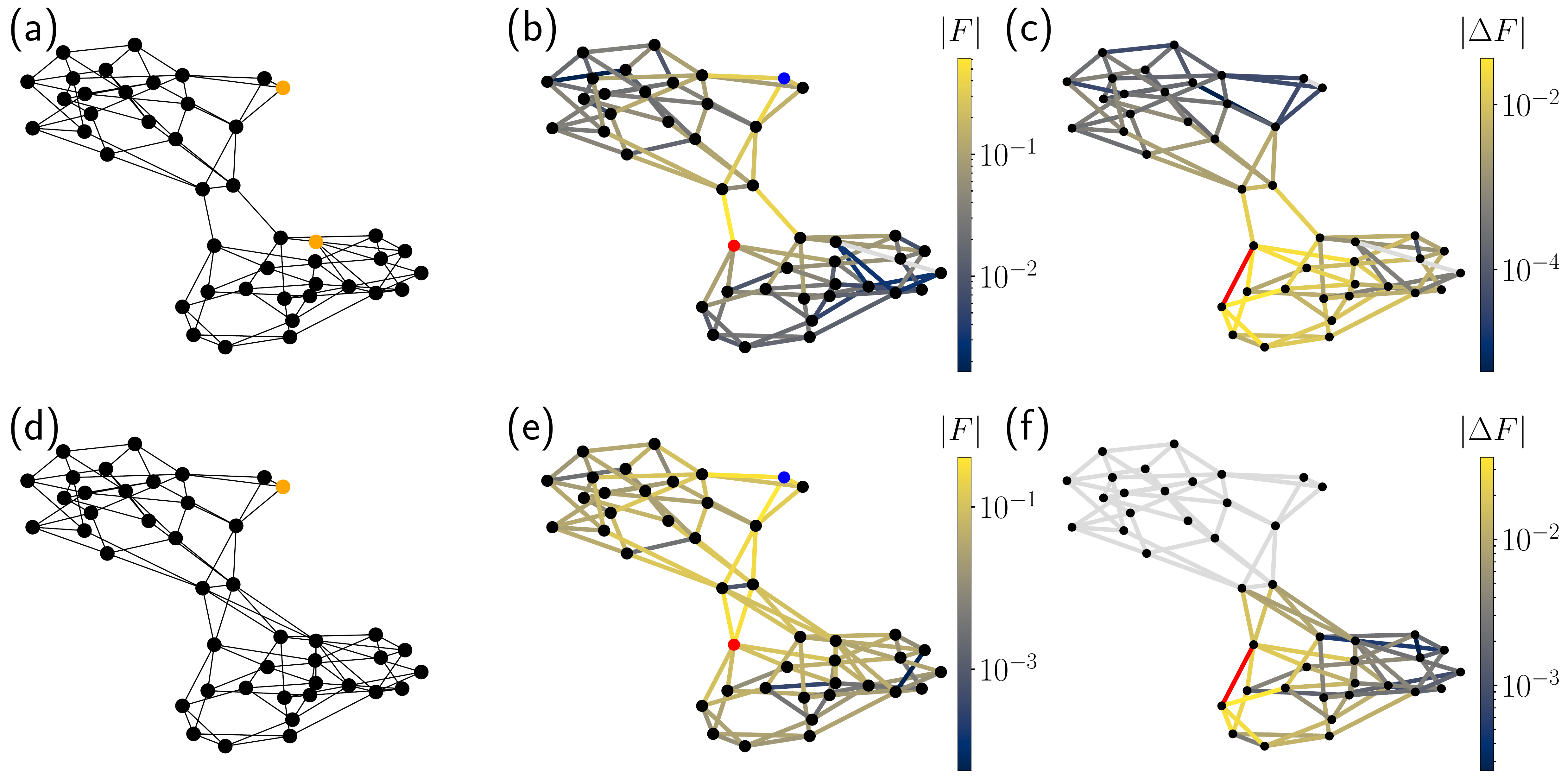}
    \end{center}
    \caption{
    \textbf{Isolators do not generally prevent the controllability of a network.} (a) An example of an undirected network with two weakly connected components that requires $N_D = 2$ driving nodes (in orange) to be controlled. This can be calculated from the graph adjacency matrix, which has, by construction, an eigenvalue $\lambda^M = -1$ with algebraic multiplicity $\delta(\lambda^M) = 2$ (See Eq.~\ref{eq:controllability_yuan} and Ref.~\cite{yuan_exact_2013}). (d) After adding a few links to create a network isolator, we have $N_D=1$ and only one node (colored orange) is necessary to control the entire network, i.e. the network isolator has in this case increased the controllability of the network. (b) We show the flows obtained by our linear flow model for a single source of power $P=1$ at the node colored in red and a single sink with $P=-1$ at the node colored in blue. The resulting (absolute) flows are color-coded: The flow can easily reach from the red node to the blue node. (e) Adding the isolator, flow can still propagate freely from the source node (red) to the target node (blue) in the same way as in panel B. Hence, the isolator does not prevent the propagation of flows. (c) Simulating the failure of a single link (red), we observe that flows do also change in the other part of the network. 
    (f) Conversely, the isolator does prevent propagation of flow changes caused by a link failure in the right part of the network to its left part.
    \label{fig:isolator_controllability}}
\end{figure*}

\end{document}